\documentclass[final,12pt]{article}
\usepackage{hyperref}
\usepackage{amssymb,amsmath,amsthm}
\usepackage[mathscr]{eucal}
\usepackage{cite}
\usepackage{tikz}
\usetikzlibrary{arrows,backgrounds,decorations.pathmorphing,
decorations.pathreplacing,positioning,fit}
\usepackage{enumerate}
\usepackage[conditional,light,first,bottomafter]{draftcopy}
\draftcopyName{DRAFT\space\today}{130}
\draftcopySetScale{65}
\usepackage[letterpaper,hmargin=3.7cm,vmargin=3.2cm]{geometry}
\usepackage{setspace}
\makeatletter
\renewcommand{\section}{\@startsection%
{section}%
{1}%
{0em}%
{1.7em}%
{1.2em}%
{\normalfont\large\centering\bfseries}}
\renewcommand{\@seccntformat}[1]%
{\csname the#1\endcsname.\hspace{0.5em}}

\numberwithin{equation}{section}
\renewcommand\appendix{\par
\setcounter{section}{0}%
\setcounter{subsection}{0}%
\setcounter{theorem}{0}
\setcounter{table}{0}
\setcounter{figure}{0}
\gdef\thetable{\Alph{table}}
\gdef\thefigure{\Alph{figure}}
\section*{Appendix}
\gdef\thesection{\Alph{section}}
\setcounter{section}{0}}

\newtheorem{theorem}{Theorem}[section]
\newtheorem{proposition}{Proposition}[section]
\newtheorem{lemma}{Lemma}[section]
\newtheorem{corollary}{Corollary}[section]
\theoremstyle{definition}
\newtheorem{definition}{Definition}
\newtheorem{remark}{Remark}[section]

\newcommand{\abs}[1]{\left|#1\right|}
\newcommand{\norm}[1]{\left\|#1\right\|}
\newcommand{\inner}[2]{\left\langle#1,#2\right\rangle}
\newcommand{\cH}{\mathcal{H}}
\newcommand{\integers}{\mathbb{Z}}
\newcommand{\tb}[1]{\widetilde{\boldsymbol{#1}}}
\newcommand{\reals}{\mathbb{R}}
\newcommand{\complex}{\mathbb{C}}
\newcommand{\nats}{\mathbb{N}}
\newcommand{\pb}[1]{\boldsymbol{#1}}

\DeclareMathOperator{\rank}{rank}
\DeclareMathOperator{\spec}{spec}

\DeclareMathOperator{\diag}{diag}
\begin{document}
\begin{titlepage}
\title{Inverse spectral analysis for a class of finite band symmetric matrices
\footnotetext{%
Mathematics Subject Classification(2010):
34K29;  
47B36. 
}
\footnotetext{%
Keywords:
Inverse spectral problem;
Band symmetric matrices;
Spectral measure.
}\hspace{-5mm}
\thanks{%
Research partially supported by UNAM-DGAPA-PAPIIT IN105414
}%
}
\author{
\textbf{Mikhail Kudryavtsev}
\\
\small Department of Mathematics\\[-1.6mm]
\small Institute for Low Temperature Physics and Engineering\\[-1.6mm]
\small Lenin Av. 47, 61103\\[-1.6mm]
\small Kharkov, Ukraine\\[-1.6mm]
\small\texttt{kudryavtsev@onet.com.ua}
\\[2mm]
\textbf{Sergio Palafox}
\\
\small Instituto de F\'isica y Matem\'aticas\\[-1.6mm]
\small Universidad Tecnol\'ogica de la Mixteca\\[-1.6mm]
\small Km. 2.5 Carr. a Acatlima, C.P. 69000, Oaxaca, M\'exico\\[-1.6mm]
\small \texttt{sergiopalafoxd@gmail.com}
\\[2mm]
\textbf{Luis O. Silva}
\\
\small Departamento de F\'{i}sica Matem\'{a}tica\\[-1.6mm]
\small Instituto de Investigaciones en Matem\'aticas Aplicadas y en Sistemas\\[-1.6mm]
\small Universidad Nacional Aut\'onoma de M\'exico\\[-1.6mm]
\small C.P. 04510, M\'exico D.F.\\[-1.6mm]
\small \texttt{silva@iimas.unam.mx}
}
\date{}
\maketitle
\vspace{-4mm}
\begin{center}
\begin{minipage}{5in}

\centerline{{\bf Abstract}}
\bigskip

In this note, we solve an inverse spectral problem for a class of
finite band symmetric matrices. We provide necessary and sufficient
conditions for a matrix valued function to be a spectral function of
the operator corresponding to a matrix in our class and give an
algorithm for recovering this matrix from the spectral function.  The
reconstructive algorithm is applicable to matrices which cannot be
treated by known inverse block matrix methods.  Our approach to the
inverse problem is based on the rational interpolation theory
developed in a previous paper.
\end{minipage}
\end{center}
\thispagestyle{empty}
\end{titlepage}
\section{Introduction}
\label{sec:intro}
This work deals with the direct and inverse spectral analysis of a
class of finite symmetric band matrices with emphasis in the inverse
problems of characterization and reconstruction. Inverse spectral
problems for band matrices have been studied extensively in the
particular case of Jacobi matrices (see for instance
\cite{MR2263317,MR504044,MR2915295,MR1616422,MR0447294,MR0213379,
  MR0382314,MR549425,MR1463594,MR1436689,MR1247178} for the finite
case and
\cite{MR2998707,MR3113459,MR1045318,MR1616422,MR499269,MR0221315,
  MR2305710,MR2438732} for the infinite case).  Works dealing with
band matrices non-necessarily tridiagonal are not so abundant (see
\cite{MR629608,MR2533388,MR2592784,MR1668981,MR1699440,MR636029,
  MR2110489,MR2432761} for the finite case and
\cite{MR2043894,MR2494240} for the infinite case).

Let $\cH$ be a finite dimensional Hilbert space with a fixed orthonormal
basis $\{\delta_k\}_{k=1}^N$ in it. For any $j=0,1,\dots,n$ with $n<
N$, consider the operator $D_j$  whose matrix representation with
respect to $\{\delta_k\}_{k=1}^N$ is a diagonal real matrix, i.\,e.,
$D_j\delta_k=d^{(j)}_k\delta_k$ for all $k=1,\dots,N$, where
$d_k^{(j)}$ is a real number. Also, let $S$ be the shift operator,
that is,
\begin{equation*}
  S\delta_k=
  \begin{cases}
    \delta_{k+1} & k=1,\dots,N-1\\
    0           & k=N\,.
  \end{cases}
\end{equation*}
The object of our considerations in this note is the symmetric
operator
\begin{equation*}
  A:=D_0+\sum_{j=1}^nS^jD_j + \sum_{j=1}^nD_j(S^*)^j\,.
\end{equation*}
Hence, the matrix representation of $A$ with respect to
$\{\delta_k\}_{k=1}^N$ is an Hermitian band matrix with real entries
which is denoted by $\mathcal{A}$. Alongside the matrix $\mathcal{A}$,
for any $j\in\{0,\dots,n\}$, we consider the diagonal matrix
$\mathcal{D}_j$ being the matrix representation with respect to
$\{\delta_k\}_{k=1}^N$ of the operator $D_j$.

We assume that the diagonals $\mathcal{D}_j$ ($j\in\{0,\dots,n\}$)
satisfy the following conditions.  The diagonal farthest from the main
one, that is
\begin{equation*}
 \mathcal{D}_n=\diag\{d_k^{(n)}\}\,,
\end{equation*}
is such that
one of the following alternatives takes place:
\begin{itemize}
\item For some $m_1<N-n+1$, all the numbers
$d_{1}^{(n)},\dots,d_{m_1-1}^{(n)}$ are strictly positive and
$d_{m_1}^{(n)}=\dots=d_{N-n}^{(n)}=0$.
\item All the elements of the sequence $\mathcal{D}_n$
are positive which we convene to mean that $m_1=N-n+1$. In this case,
we define $m_j=N-n+j$ for $j=2,\dots,n$.
\end{itemize}

If $m_1< N-n+1$, we say that $\mathcal{D}_n$ undergoes a degeneration
at $m_1$ and, then, the elements
$d_{m_1+1}^{(n-1)},\dots,d_{N-n+1}^{(n-1)}$ of $\mathcal{D}_{n-1}$
behave in the same way as the elements of $\mathcal{D}_n$, that is,
one of the following alternatives takes place:
\begin{itemize}
\item For some $m_2$ such that $m_1<m_2< N-n+2$,
  $d_{m_1+1}^{(n-1)},\ldots,d_{m_2-1}^{(n-1)}>0$ and
$d_{m_2}^{(n-1)}=\ldots=d_{N-n+1}^{(n-1)}=0$.
\item $d_k^{(n-1)}>0$ for $k=m_1+1,\dots,N-n+1$ and we define
  $m_2=N-n+2$ and $m_j=N-n+j$ for $j=3,\dots,n$.
\end{itemize}

We continue applying this rule up to some $j_0\le n-1$ such that
$m_{j_0}<N-n+j_0$ and $m_{j_0+1}=N-n+j_0+1$.  Finally, we define
$m_j=N-n+j$ for $j=j_0+2,\dots,n$.

Note that if one assumes that $m_j+1<m_{j+1}$ for $j=0,\dots,j_0-1$, i.\,e.,
there are no consecutive degenerations, then the elements of
$\mathcal{D}_{n-j}$ satisfy
\begin{align*}
  d_{m_j+1}^{(n-j)},\dots,d_{m_{j+1}-1}^{(n-j)}>0\,,\\
  d_{m_{j+1}}^{(n-j)}=\dots=d_{N-n+j}^{(n-j)}=0
\end{align*}
for $j=0,\dots,j_0-1$, with $m_0=0$. The entries of $\mathcal{D}_{n-j_0}$ satisfy
\begin{equation*}
 d_{m_{j_0}+1}^{(n-j_0)},\ldots,d_{N-n+j_0}^{(n-j_0)}>0\,.
\end{equation*}
The diagonal $\mathcal{D}_{n-j}$ undergoes a degeneration
at $m_{j+1}$ for $j\in\{0,\dots,j_0-1\}$. When $j_0=0$, there is no
degeneration of the diagonal $\mathcal{D}_{n}$. Note that
$\mathcal{D}_2$ is the innermost diagonal where a degeneration may
occur. Observe also that, in all cases, one has the set
of numbers $\{m_1,\dots,m_n\}$.

\begin{definition}
\label{def:matrices-degenerate}
Fix the natural numbers $n$ and $N$ such that $n<N$. All the matrices
satisfying the above properties for a given set of numbers
$\{m_i\}_{i=1}^n$ are denoted by $\mathcal{M}(n,N)$. Note that in this
notation, $N$ represents the dimension and $2n+1$ is the number of
diagonals of the matrices.
\end{definition}

\begin{figure}[h]
\begin{center}
\begin{tikzpicture}[scale=.18]\footnotesize
  \pgfmathsetmacro{\xone}{0} \pgfmathsetmacro{\xtwo}{ 30}
  \pgfmathsetmacro{\yone}{0} \pgfmathsetmacro{\ytwo}{30}
  \draw[step=1cm,gray,opacity=0.5,very thin] (\xone,\yone) grid
  (\xtwo,\ytwo); 
  \draw[step=1cm,gray,opacity=0.5,very thin] (35,7)
  grid (36,8); 
  \draw[step=1cm,gray,opacity=0.5,very thin] (35,5) grid
  (36,6); 
  \draw[step=1cm,gray,opacity=0.5,very thin] (35,7) grid
  (36,8); 
  \draw(46,7.5)node[text width=3cm,scale=1]{zeros}; 
  \draw(46,5.7)node[text width=3cm,scale=1]{real numbers}; 
  \draw(46,3.5)node[text width=3cm,scale=1]{positive numbers}; 
  \draw(43,22)node[text width=3cm,scale=1]{\bf degenerations}; 
  \draw(44,17.1)node[text width=4.2cm,align=justify, scale=.8] {(the dashed 
 arrow points to consecutive degenerations, i.\,e., when $m_{j+1}=m_j+1$)}; 
  \draw[->] (33.5,22) -- (16,21.5); 
  \draw[->] (33.5,22) -- (13,25.5); 
  \draw[dashed,->] (33.5,22) -- (17,19); 
  \draw[->] (33.5,22) -- (19,15.5); 
  \draw[->] (33.5,22) -- (23,11); 
  \draw[->] (33.5,22) -- (25,8);


\begin{scope}
  \filldraw[thin,gray,opacity=.4] (35,5)
    rectangle (36,6)
 ;
  \filldraw[thin,gray,opacity=1] (35,3)
    rectangle (36,4);
\end{scope}
\begin{scope}
\foreach \x in {0,1,2,3}
{
  \filldraw[thin,gray,opacity=1] (0+\x, 21-\x)
    rectangle (1+\x,22-\x)
 ;
   \filldraw[thin,gray,opacity=1] (8+\x, 30-\x)
     rectangle (9+\x,29-\x);}
\end{scope}
\begin{scope}
\foreach \x in {0,1,2,3,4}
{
  \filldraw[thin,gray,opacity=.4] (0+\x, 22-\x)
    rectangle (1+\x,23-\x)
 ;
   \filldraw[thin,gray,opacity=.4] (7+\x, 30-\x)
     rectangle (8+\x,29-\x);}
\end{scope}

\begin{scope}
\foreach \x in {5,6,7}
{
  \filldraw[thin,gray,opacity=1] (0+\x, 22-\x)
    rectangle (1+\x,23-\x)
 ;
   \filldraw[thin,gray,opacity=1] (7+\x, 30-\x)
     rectangle (8+\x,29-\x);}
\end{scope}
\begin{scope}
\foreach \x in {0,1,2,3,4,5,6,7,8}
{
  \filldraw[thin,gray,opacity=.4] (0+\x, 23-\x)
    rectangle (1+\x,24-\x)
 ;
   \filldraw[thin,gray,opacity=.4] (6+\x, 30-\x)
     rectangle (7+\x,29-\x);}
\end{scope}
\begin{scope}
\foreach \x in {9}
{
  \filldraw[thin,gray,opacity=1] (0+\x, 23-\x)
    rectangle (1+\x,24-\x)
 ;
   \filldraw[thin,gray,opacity=1] (6+\x, 30-\x)
     rectangle (7+\x,29-\x);}
\end{scope}
\begin{scope}
\foreach \x in {0,1,2,3,4,5,6,7,8,9,10}
{
  \filldraw[thin,gray,opacity=.4] (0+\x, 24-\x)
    rectangle (1+\x,25-\x)
 ;
   \filldraw[thin,gray,opacity=.4] (5+\x, 30-\x)
     rectangle (6+\x,29-\x);}
\end{scope}

\begin{scope}
\foreach \x in {0,1,2,3,4,5,6,7,8,9,10,11,12}
{
  \filldraw[thin,gray,opacity=.4] (0+\x, 25-\x)
    rectangle (1+\x,26-\x)
 ;
   \filldraw[thin,gray,opacity=.4] (4+\x, 30-\x)
     rectangle (5+\x,29-\x);}
\end{scope}

\begin{scope}
\foreach \x in {12,13}
{
  \filldraw[thin,gray,opacity=1] (0+\x, 25-\x)
    rectangle (1+\x,26-\x)
 ;
   \filldraw[thin,gray,opacity=1] (4+\x, 30-\x)
     rectangle (5+\x,29-\x);}
\end{scope}
\begin{scope}
\foreach \x in {0,1,2,3,4,5,6,7,8,9,10,11,12,13,14}
{
  \filldraw[thin,gray,opacity=.4] (0+\x, 26-\x)
    rectangle (1+\x,27-\x)
 ;
   \filldraw[thin,gray,opacity=.4] (3+\x, 30-\x)
     rectangle (4+\x,29-\x);}
\end{scope}

\begin{scope}
\foreach \x in {15,16,17,18}
{
  \filldraw[thin,gray,opacity=1] (0+\x, 26-\x)
    rectangle (1+\x,27-\x)
 ;
   \filldraw[thin,gray,opacity=1] (3+\x, 30-\x)
     rectangle (4+\x,29-\x);}
\end{scope}

\begin{scope}
\foreach \x in {0,1,2,3,4,5,6,7,8,9,10,11,12,13,14,15,16,17,18,19}
{
  \filldraw[thin,gray,opacity=.4] (0+\x, 27-\x)
    rectangle (1+\x,28-\x)
 ;
   \filldraw[thin,gray,opacity=.4] (2+\x, 30-\x)
     rectangle (3+\x,29-\x);}
\end{scope}

\begin{scope}
\foreach \x in {20,21}
{
  \filldraw[thin,gray,opacity=1] (0+\x, 27-\x)
    rectangle (1+\x,28-\x)
 ;
   \filldraw[thin,gray,opacity=1] (2+\x, 30-\x)
     rectangle (3+\x,29-\x);}
\end{scope}
\begin{scope}
\foreach \x in {0,1,2,3,4,5,6,7,8,9,10,11,12,13,14,15,
16,17,18,19,20,21,22}
{
  \filldraw[thin,gray,opacity=.4] (0+\x, 28-\x)
    rectangle (1+\x,29-\x)
 ;
   \filldraw[thin,gray,opacity=.4] (1+\x, 30-\x)
     rectangle (2+\x,29-\x);}
\end{scope}

\begin{scope}
\foreach \x in {23,24,25,26,27,28}
{
  \filldraw[thin,gray,opacity=1] (0+\x, 28-\x)
    rectangle (1+\x,29-\x)
 ;
   \filldraw[thin,gray,opacity=1] (1+\x, 30-\x)
     rectangle (2+\x,29-\x);}
\end{scope}

\begin{scope}
  \foreach \x in
  {0,1,2,3,4,5,6,7,8,9,10,11,12,13,14,15,16,17,18,19,
20,21,22,23,24,25,26,27,28,29}
  { \filldraw[thin,gray,opacity=.25] (0+\x, 29-\x) rectangle
    (1+\x,30-\x) ; \filldraw[thin,gray,opacity=.2] (0+\x, 29-\x)
    rectangle (1+\x,30-\x);}
\end{scope}
\end{tikzpicture}
\end{center}
\caption{The structure of a matrix in $\mathcal{M}(n,N)$}\label{fig:structure-class}
\end{figure}
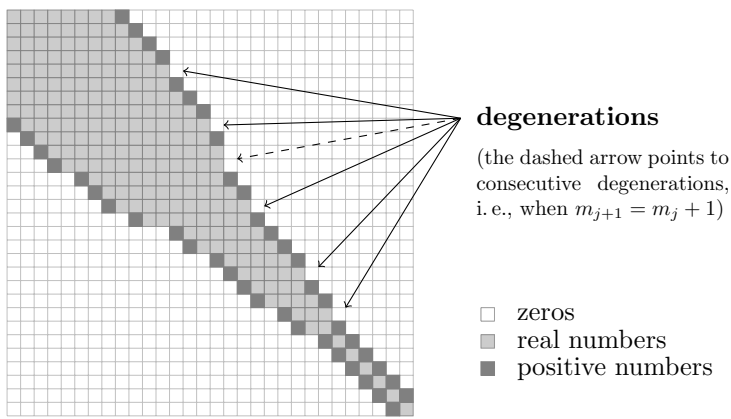


An example of a matrix in $\mathcal{M}(3,7)$, when $m_1=3$, $m_2=5$
and $m_3=7$ is the following.
\begin{equation*}
\mathcal{A}=
\begin{small}
  \begin{pmatrix}
    d^{(0)}_1&d^{(1)}_1&d^{(2)}_1&d^{(3)}_1&0&0&0\\[1mm]
    d^{(1)}_1&d^{(0)}_2&d^{(1)}_2&d^{(2)}_2&d^{(3)}_2&0&0\\[1mm]
    d^{(2)}_1&d^{(1)}_2&d^{(0)}_3&d^{(1)}_3&d^{(2)}_3&0&0\\[1mm]
    d^{(3)}_1&d^{(2)}_2&d^{(1)}_3&d^{(0)}_4&d^{(1)}_4&d^{(2)}_4&0\\[1mm]
    0&d^{(3)}_2&d^{(2)}_3&d^{(1)}_4&d^{(0)}_5&d^{(1)}_5&0\\[1mm]
    0&0&0&d^{(2)}_4&d^{(1)}_5&d^{(0)}_6&d^{(1)}_6\\[1mm]
    0&0&0&0&0&d^{(1)}_{6}&d^{(0)}_7
  \end{pmatrix}\,.
\end{small}
\end{equation*}
Here we say that the matrix $\mathcal{A}$ underwent a degeneration of
the diagonal $\mathcal{D}_3$ in $m_1=3$ and a degeneration of
$\mathcal{D}_2$ in $m_2=5$. Observe that $j_0=2$ and that
$d_1^{(3)},\ d_2^{(3)},\ d_{4}^{(2)},\ d_6^{(1)}$  are positive numbers.

It is known that the dynamics of a finite linear mass-spring system is
characterized by the spectral properties of a finite Jacobi matrix
\cite{MR2102477,mono-marchenko} (see Figure~\ref{fig:0}) when the system
is within the regime of validity of the Hooke law. The entries of the
Jacobi matrix are determined by the masses and spring constants of the
system
\cite{MR2915295,MR2998707,MR3113459,MR2102477,mono-marchenko}. The
movement of the mechanical system of Figure~\ref{fig:0} is a
superposition of harmonic oscillations whose frequencies are given by
the spectrum of the Jacobi operator.
\begin{figure}[h]
\begin{center}
\resizebox{1\textwidth}{!}{%
\begin{tikzpicture}
  [mass1/.style={circle,draw=black!80,fill=black!13,thick,inner sep=0pt,
   minimum size=5mm},
   mass2/.style={circle,draw=black!80,fill=black!13,thick,inner sep=0pt,
   minimum size=3.7mm},
   mass3/.style={circle,draw=black!80,fill=black!13,thick,inner sep=0pt,
   minimum size=5.7mm},
   mass4/.style={circle,draw=black!80,fill=black!13,thick,inner sep=0pt,
   minimum size=5mm},
   mass5/.style={circle,draw=black!80,fill=black!13,thick,inner sep=0pt,
   minimum size=4mm},
   mass6/.style={circle,draw=black!80,fill=black!13,thick,inner sep=0pt,
   minimum size=5.2mm},
   mass7/.style={circle,draw=black!80,fill=black!13,thick,inner sep=0pt,
   minimum size=6mm},
   mass8/.style={circle,draw=black!80,fill=black!13,thick,inner sep=0pt,
   minimum size=5.2mm},
   mass9/.style={circle,draw=black!80,fill=black!13,thick,inner sep=0pt,
   minimum size=5.4mm},
   massn/.style={circle,draw=black!80,fill=black!13,thick,inner sep=0pt,
   minimum size=5.2mm},
   wall/.style={postaction={draw,decorate,decoration={border,angle=-45,
   amplitude=0.3cm,segment length=1.5mm}}},
   wall1/.style={postaction={draw,decorate,decoration={border,angle=45,
   amplitude=0.3cm,segment length=1.5mm}}}]
  \node (massn) at (12.75,1) [massn] {};
  \node (mass9) at (11.5,1) [mass8] {};
  \node (mass8) at (10.25,1) [mass8] {};
  \node (mass7) at (9.0,1) [mass7] {};
  \node (mass6) at (7.75,1) [mass6] {};
  \node (mass5) at (6.5,1) [mass5] {};
  \node (mass4) at (5.25,1) [mass4] {};
  \node (mass3) at (4.0,1) [mass3] {};
  \node (mass2) at (2.75,1) [mass2] {};
  \node (mass1) at (1.5,1) [mass1] {};
\draw[decorate,decoration={coil,aspect=0.4,segment
  length=1.1mm,amplitude=0.7mm}] (0.5,1) -- node[below=4pt]
{} (mass1);
\draw[decorate,decoration={coil,aspect=0.4,segment
  length=1.4mm,amplitude=0.7mm}] (mass1) -- node[below=4pt]
{} (mass2);
\draw[decorate,decoration={coil,aspect=0.4,segment
  length=1.5mm,amplitude=0.7mm}] (mass2) -- node[below=4pt]
{} (mass3);
\draw[decorate,decoration={coil,aspect=0.4,segment
  length=1.1mm,amplitude=0.7mm}] (mass3) -- node[below=4pt]
{} (mass4);
\draw[decorate,decoration={coil,aspect=0.4,segment
  length=0.9mm,amplitude=0.7mm}] (mass4) -- node[below=4pt]
{} (mass5);
\draw[decorate,decoration={coil,aspect=0.4,segment
  length=1.4mm,amplitude=0.7mm}] (mass5) -- node[below=4pt]
{} (mass6);
\draw[decorate,decoration={coil,aspect=0.4,segment
  length=1.7mm,amplitude=0.7mm}] (mass6) -- node[below=4pt]
{} (mass7);
\draw[decorate,decoration={coil,aspect=0.4,segment
  length=0.8mm,amplitude=0.7mm}] (massn) -- node[below=4pt]
{} (13.75,1);
\draw[decorate,decoration={coil,aspect=0.4,segment
  length=1.1mm,amplitude=0.7mm}] (mass7) -- node[below=4pt]
{} (mass8);
\draw[decorate,decoration={coil,aspect=0.4,segment
  length=1.3mm,amplitude=0.7mm}] (mass8) -- node[below=4pt]
{} (mass9);
\draw[decorate,decoration={coil,aspect=0.4,segment
  length=1.7mm,amplitude=0.7mm}] (mass9) -- node[below=4pt]
{} (massn);
\draw[line width=.5pt,wall](0.5,1.7)--(0.5,0.3);
\draw[line width=.5pt,wall1](13.75,1.7)--(13.75,0.3);
\end{tikzpicture}
}
\end{center}
\caption{Mass-spring system corresponding to a Jacobi
  matrix}\label{fig:0}
\end{figure}
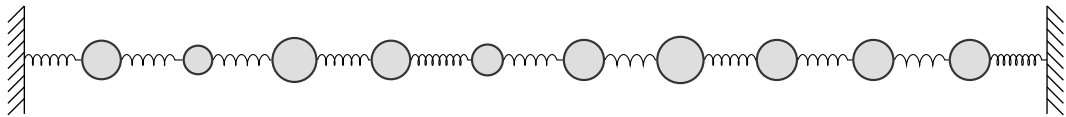
Analogously, one can deduce that a matrix in $\mathcal{M}(n,N)$ models
a linear mass-spring system where the interaction extends to all the
$n$ neighbors of each mass (see Appendix~\ref{sec:Mass-spring}). For instance, if the matrix is in
$\mathcal{M}(2,10)$ and no degeneration of the diagonals occurs,
viz. $m_1=9$, the corresponding mass-spring system is given in
Figure~\ref{fig:1}.
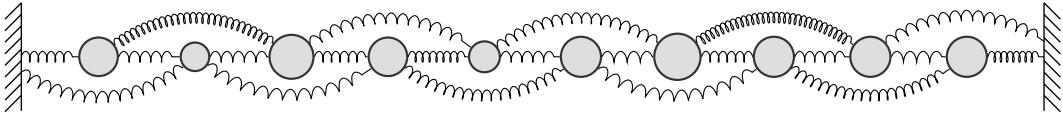
\begin{figure}[h]
\begin{center}
\resizebox{1\textwidth}{!}{%
\begin{tikzpicture}
  [mass1/.style={circle,draw=black!80,fill=black!13,thick,inner sep=0pt,
   minimum size=5mm},
   mass2/.style={circle,draw=black!80,fill=black!13,thick,inner sep=0pt,
   minimum size=3.7mm},
   mass3/.style={circle,draw=black!80,fill=black!13,thick,inner sep=0pt,
   minimum size=5.7mm},
   mass4/.style={circle,draw=black!80,fill=black!13,thick,inner sep=0pt,
   minimum size=5mm},
   mass5/.style={circle,draw=black!80,fill=black!13,thick,inner sep=0pt,
   minimum size=4mm},
   mass6/.style={circle,draw=black!80,fill=black!13,thick,inner sep=0pt,
   minimum size=5.2mm},
   mass7/.style={circle,draw=black!80,fill=black!13,thick,inner sep=0pt,
   minimum size=6mm},
   mass8/.style={circle,draw=black!80,fill=black!13,thick,inner sep=0pt,
   minimum size=5.2mm},
   mass9/.style={circle,draw=black!80,fill=black!13,thick,inner sep=0pt,
   minimum size=5.4mm},
   massn/.style={circle,draw=black!80,fill=black!13,thick,inner sep=0pt,
   minimum size=5.2mm},
   wall/.style={postaction={draw,decorate,decoration={border,angle=-45,
   amplitude=0.3cm,segment length=1.5mm}}},
   wall1/.style={postaction={draw,decorate,decoration={border,angle=45,
   amplitude=0.3cm,segment length=1.5mm}}}]
  \node (massn) at (12.75,1) [massn] {};
  \node (mass9) at (11.5,1) [mass8] {};
  \node (mass8) at (10.25,1) [mass8] {};
  \node (mass7) at (9.0,1) [mass7] {};
  \node (mass6) at (7.75,1) [mass6] {};
  \node (mass5) at (6.5,1) [mass5] {};
  \node (mass4) at (5.25,1) [mass4] {};
  \node (mass3) at (4.0,1) [mass3] {};
  \node (mass2) at (2.75,1) [mass2] {};
  \node (mass1) at (1.5,1) [mass1] {};
\draw[decorate,decoration={coil,aspect=0.4,segment
  length=1.1mm,amplitude=0.7mm}] (0.5,1) -- node[below=4pt]
{} (mass1);
\draw[decorate,decoration={coil,aspect=0.4,segment
  length=1.4mm,amplitude=0.7mm}] (mass1) -- node[below=4pt]
{} (mass2);
\draw[decorate,decoration={coil,aspect=0.4,segment
  length=1.5mm,amplitude=0.7mm}] (mass2) -- node[below=4pt]
{} (mass3);
\draw[decorate,decoration={coil,aspect=0.4,segment
  length=1.1mm,amplitude=0.7mm}] (mass3) -- node[below=4pt]
{} (mass4);
\draw[decorate,decoration={coil,aspect=0.4,segment
  length=0.9mm,amplitude=0.7mm}] (mass4) -- node[below=4pt]
{} (mass5);
\draw[decorate,decoration={coil,aspect=0.4,segment
  length=1.4mm,amplitude=0.7mm}] (mass5) -- node[below=4pt]
{} (mass6);
\draw[decorate,decoration={coil,aspect=0.4,segment
  length=1.7mm,amplitude=0.7mm}] (mass6) -- node[below=4pt]
{} (mass7);
\draw[decorate,decoration={coil,aspect=0.4,segment
  length=0.8mm,amplitude=0.7mm}] (massn) -- node[below=4pt]
{} (13.75,1);
\draw[decorate,decoration={coil,aspect=0.4,segment
  length=1.1mm,amplitude=0.7mm}] (mass7) -- node[below=4pt]
{} (mass8);
\draw[decorate,decoration={coil,aspect=0.4,segment
  length=1.3mm,amplitude=0.7mm}] (mass8) -- node[below=4pt]
{} (mass9);
\draw[decorate,decoration={coil,aspect=0.4,segment
  length=1.7mm,amplitude=0.7mm}] (mass9) -- node[below=4pt]
{} (massn);
\draw[decorate,decoration={coil,aspect=0.4,segment
  length=0.8mm,amplitude=0.7mm}] (mass1) to [bend left=35] node[below=4pt]
{} (mass3);
\draw[decorate,decoration={coil,aspect=0.4,segment
  length=1.5mm,amplitude=0.7mm}] (mass3) to [bend left=35] node[below=4pt]
{} (mass5);
\draw[decorate,decoration={coil,aspect=0.4,segment
  length=1.3mm,amplitude=0.7mm}] (mass5) to [bend left=35] node[below=4pt]
{} (mass7);
\draw[decorate,decoration={coil,aspect=0.4,segment
  length=0.7mm,amplitude=0.7mm}] (mass7) to [bend left=35] node[below=4pt]
{} (mass9);
\draw[decorate,decoration={coil,aspect=0.4,segment
  length=1.5mm,amplitude=0.7mm}] (mass9) to [bend left=35] node[below=4pt]
{} (13.75,1.2);
\draw[decorate,decoration={coil,aspect=0.4,segment
  length=1.5mm,amplitude=0.7mm}] (0.5,0.8) to [bend right=35] node[below=4pt]
{} (mass2);
\draw[decorate,decoration={coil,aspect=0.4,segment
  length=1.8mm,amplitude=0.7mm}] (mass2) to [bend right=35] node[below=4pt]
{} (mass4);
\draw[decorate,decoration={coil,aspect=0.4,segment
  length=1.2mm,amplitude=0.7mm}] (mass4) to [bend right=35] node[below=4pt]
{} (mass6);
\draw[decorate,decoration={coil,aspect=0.4,segment
  length=1.8mm,amplitude=0.7mm}] (mass6) to [bend right=35] node[below=4pt]
{} (mass8);
\draw[decorate,decoration={coil,aspect=0.4,segment
  length=1.1mm,amplitude=0.7mm}] (mass8) to [bend right=35] node[below=4pt]
{} (massn);
\draw[line width=.5pt,wall](0.5,1.7)--(0.5,0.3);
\draw[line width=.5pt,wall1](13.75,1.7)--(13.75,0.3);
\end{tikzpicture}
}
\end{center}
\caption{Mass-spring system of a matrix in
  $\mathcal{M}(2,10)$: nondegenerated case}\label{fig:1}
\end{figure}
If for another matrix in $\mathcal{M}(2,10)$, one has degeneration of the
diagonals, for instance $m_1=4$, the corresponding mass-spring system
is given in Figure~\ref{fig:2}.
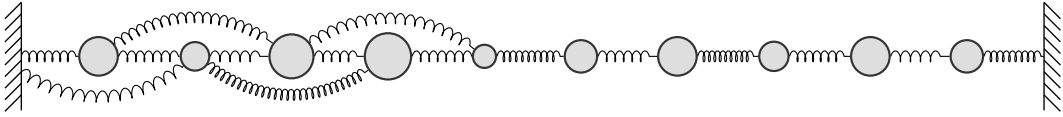
\begin{figure}[h]
\begin{center}
\resizebox{1\textwidth}{!}{%
\begin{tikzpicture}
  [mass1/.style={circle,draw=black!80,fill=black!13,thick,inner sep=0pt,
   minimum size=5mm},
   mass2/.style={circle,draw=black!80,fill=black!13,thick,inner sep=0pt,
   minimum size=3.7mm},
   mass3/.style={circle,draw=black!80,fill=black!13,thick,inner sep=0pt,
   minimum size=5.7mm},
   mass4/.style={circle,draw=black!80,fill=black!13,thick,inner sep=0pt,
   minimum size=6mm},
   mass5/.style={circle,draw=black!80,fill=black!13,thick,inner sep=0pt,
   minimum size=3mm},
   mass6/.style={circle,draw=black!80,fill=black!13,thick,inner sep=0pt,
   minimum size=4.2mm},
   mass7/.style={circle,draw=black!80,fill=black!13,thick,inner sep=0pt,
   minimum size=5mm},
   mass8/.style={circle,draw=black!80,fill=black!13,thick,inner sep=0pt,
   minimum size=3.8mm},
   mass9/.style={circle,draw=black!80,fill=black!13,thick,inner sep=0pt,
   minimum size=5mm},
   massn/.style={circle,draw=black!80,fill=black!13,thick,inner sep=0pt,
   minimum size=4.2mm},
   wall/.style={postaction={draw,decorate,decoration={border,angle=-45,
   amplitude=0.3cm,segment length=1.5mm}}},
   wall1/.style={postaction={draw,decorate,decoration={border,angle=45,
   amplitude=0.3cm,segment length=1.5mm}}}]
  \node (massn) at (12.75,1) [massn] {};
  \node (mass9) at (11.5,1) [mass9] {};
  \node (mass8) at (10.25,1) [mass8] {};
  \node (mass7) at (9.0,1) [mass7] {};
  \node (mass6) at (7.75,1) [mass6] {};
  \node (mass5) at (6.5,1) [mass5] {};
  \node (mass4) at (5.25,1) [mass4] {};
  \node (mass3) at (4.0,1) [mass3] {};
  \node (mass2) at (2.75,1) [mass2] {};
  \node (mass1) at (1.5,1) [mass1] {};
\draw[decorate,decoration={coil,aspect=0.4,segment
  length=1.0mm,amplitude=0.7mm}] (0.5,1) -- node[below=4pt]
{} (mass1);
\draw[decorate,decoration={coil,aspect=0.4,segment
  length=1.1mm,amplitude=0.7mm}] (mass1) -- node[below=4pt]
{} (mass2);
\draw[decorate,decoration={coil,aspect=0.4,segment
  length=1.3mm,amplitude=0.7mm}] (mass2) -- node[below=4pt]
{} (mass3);
\draw[decorate,decoration={coil,aspect=0.4,segment
  length=1.1mm,amplitude=0.7mm}] (mass3) -- node[below=4pt]
{} (mass4);
\draw[decorate,decoration={coil,aspect=0.4,segment
  length=1.1mm,amplitude=0.7mm}] (mass4) -- node[below=4pt]
{} (mass5);
\draw[decorate,decoration={coil,aspect=0.4,segment
  length=0.8mm,amplitude=0.7mm}] (mass5) -- node[below=4pt]
{} (mass6);
\draw[decorate,decoration={coil,aspect=0.4,segment
  length=1.1mm,amplitude=0.7mm}] (mass6) -- node[below=4pt]
{} (mass7);
\draw[decorate,decoration={coil,aspect=0.4,segment
  length=1.0mm,amplitude=0.7mm}] (mass8) -- node[below=4pt]
{} (mass9);
\draw[decorate,decoration={coil,aspect=0.4,segment
  length=0.8mm,amplitude=0.7mm}] (massn) -- node[below=4pt]
{} (13.75,1);
\draw[decorate,decoration={coil,aspect=0.4,segment
  length=0.7mm,amplitude=0.7mm}] (mass7) -- node[below=4pt]
{} (mass8);
\draw[decorate,decoration={coil,aspect=0.4,segment
  length=1.3mm,amplitude=0.7mm}] (mass9) -- node[below=4pt]
{} (massn);
\draw[decorate,decoration={coil,aspect=0.4,segment
  length=1.1mm,amplitude=0.7mm}] (mass1) to [bend left=35] node[below=4pt]
{} (mass3);
\draw[decorate,decoration={coil,aspect=0.4,segment
  length=1.3mm,amplitude=0.7mm}] (mass3) to [bend left=35] node[below=4pt]
{} (mass5);
\draw[decorate,decoration={coil,aspect=0.4,segment
  length=1.5mm,amplitude=0.7mm}] (0.5,0.8) to [bend right=35] node[below=4pt]
{} (mass2);
\draw[decorate,decoration={coil,aspect=0.4,segment
  length=0.8mm,amplitude=0.7mm}] (mass2) to [bend right=35] node[below=4pt]
{} (mass4);
\draw[line width=.5pt,wall](0.5,1.7)--(0.5,0.3);
\draw[line width=.5pt,wall1](13.75,1.7)--(13.75,0.3);
\end{tikzpicture}
}
\end{center}
\caption{Mass-spring system of a matrix in
  $\mathcal{M}(2,10)$: degenerated case}\label{fig:2}
\end{figure}

In this work, the approach to the inverse spectral analysis of the
operators whose matrix representation belongs to $\mathcal{M}(n,N)$ is
based on the one used in \cite{MR1668981,MR1699440}, but it allows to
treat the case of arbitrary $n$. An important ingredient of the method
used here is the linear interpolation of $n$-dimensional vector
polynomials, recently developed in \cite{MR3389906}.  The linear
interpolation theory of \cite{MR3389906} is a nontrivial
generalization of the rational interpolation theory developed in
\cite{MR2533388} from ideas given in \cite{MR1668981,MR1699440}.  It
is on the basis of the results of \cite{MR3389906} that the inverse
spectral theory developed in \cite{MR1668981,MR1699440} is extended
here to band matrices with $2n+1$ diagonals ($n\in\nats$). This
generalization required some new developments, in particular, we
modified the technique used in the reconstruction of the matrices in
the class $\mathcal{M}(n,N)$ which was concocted having uniqueness of
the reconstruction in mind. Indeed, the class $\mathcal{M}(2,N)$ turns
out to be a subclass of the class studied in
\cite{MR1668981,MR1699440} for which uniqueness of the reconstruction
is guaranteed.

The main results of this paper are:
\begin{enumerate}[(A)]
\item\label{it:introduction-a} A complete characterization of the
  spectral functions (measures) of operators whose matrix
  representation with respect to some orthonormal basis is an element
  of $\mathcal{M}(n,N)$
  (Proposition~\ref{prop:properties-spectral-function}).
\item\label{it:introduction-b} A characterization for degenerations of
  the diagonals of $\mathcal{A}$ in terms of the presence of
  polynomials of zero norm in $L_2(\reals,\sigma)$, where $\sigma$ is
  the spectral measure of the operator $A$.
\item\label{it:introduction-c} An algorithm for reconstructing the
  matrix from the corresponding spectral measure
  (Section~\ref{sec:reconstruction}).
\item\label{it:introduction-d} Uniqueness of reconstruction (see
  Theorem~\ref{thm:sigma-unique})
\end{enumerate}
After finishing the first version of this paper, we learned about
\cite{marchenko-slavin} (see Acknowledgments). This book considers
inverse spectral problems for a class of band symmetric finite
matrices and introduces the notion of \emph{restricted spectral data}
for which any matrix-valued spectral function in
(\ref{it:introduction-a}) is a particular case. Some of the results of
\cite{marchenko-slavin} concerning this matter are based on
\cite{MR3087910}.
Theorem 8 in \cite[Sec.\,15]{marchenko-slavin} gives necessary and
sufficient conditions for the matrix to have restricted spectral data
from which the matrix can be recovered. Moreover, from the results of
\cite[Sec.\,14]{marchenko-slavin}, one knows when the reconstruction
is unique. Theorem 9 in \cite[Sec.\,15]{marchenko-slavin}
characterizes the restricted spectral data when the matrix can be
reconstructed from them. The class $\mathcal{M}(n,N)$ satisfies
\cite[Thm.\,8 Sec.\,15]{marchenko-slavin} and the distribution of
positive numbers in Figure~\ref{fig:structure-class} guarantees that
the matrix is not ``decomposable'' \cite[Sec.\,12]{marchenko-slavin}
and satisfies the conditions for a unique reconstruction
\cite[Sec.\,14]{marchenko-slavin}.  Thus, some of our results, and
particularly (\ref{it:introduction-d}), are actually contained in
\cite{marchenko-slavin}. It is worth remarking that
(\ref{it:introduction-a}) is a different characterization of a
particular case of restricted spectral data from which the matrix can
be reconstructed. The algorithm of recontruction
(\ref{it:introduction-c}) differs from the one of
\cite{marchenko-slavin} and, remarkably, it can be applied without
essential modifications to spectral functions with infinitely many
points of increase (see \cite{MR3543793}). Finally, (\ref{it:introduction-b}) has no anlogue
in \cite{marchenko-slavin} and introduces \emph{inner} boundary
conditions (see Section~\ref{sec:spectral-measure})

It is known that if there exists a natural number $K$ such that
$N=Kn$, then a band matrix with $2n+1$ diagonals can be reduced to a
tridiagonal block matrix. However, the spectral theory for tridiagonal
block matrices requires that the off-diagonal block matrices be
invertible.  The matrices in $\mathcal{M}(n,N)$ do not satisfy this
requirement when there is a degeneration of the diagonals. The
technique for recovering a matrix from its spectral function developed
in this paper is applicable to any element in $\mathcal{M}(n,N)$ even
when $N$ is not an integer multiple of $n$.


This paper is organized as follows. The next section deals with the
direct spectral analysis of the operators under consideration. In this
section, a family of spectral functions is constructed for each
element in $\mathcal{M}(n,N)$. In
Section~\ref{sec:direct-spectral-analysis}, the connection of the
spectral analysis and the interpolation problem is
established. Section~\ref{sec:reconstruction} treats the problem of
reconstruction and characterization. In
Section~\ref{sec:alternative-methods}, we discuss alternative
approaches to the inverse spectral problem and give a comparative
analysis with the method given in
Section~\ref{sec:reconstruction}. The
Appendix~\ref{sec:Mass-spring} gives a brief account of how to
deduce the band symmetric matrix associated with a mass-spring system
from the dynamical equations.

\section{The spectral function}
\label{sec:spectral-measure}
Consider $\varphi=\sum_{k=1}^N\varphi_k\delta_k\in\cH$ and the equation
\begin{equation}
  \label{eq:eigenvectors}
  (A-z I)\varphi=0, \; z\in\mathbb{C}\,.
\end{equation}
We know that the equation has nontrivial solutions only for a finite
set of $z$.

From (\ref{eq:eigenvectors}) one obtains a system of $N$ equations,
where each equation, given by a fixed $k\in\{1,\dots,N\}$, is of
the form
\begin{equation}
\label{eq:recurrence}
\sum_{i=0}^{n-1} d_{k-n+i}^{(n-i)}\varphi_{k-n+i}+
d_k^{(0)}\varphi_k+
\sum_{i=1}^nd_k^{(i)}\varphi_{k+i}=z \varphi_k\,,
\end{equation}
where it has been assumed that
\begin{subequations}
  \label{eq:first-boundary-conditions}
\begin{align}
\label{eq:first-boundary-conditions-a}
\varphi_{k}=0\,,\quad\text{for}\ k&<1\,,\\
\label{eq:first-boundary-conditions-b}
\varphi_{k}=0\,,\quad\text{for}\ k&>N\,.
\end{align}
\end{subequations}

One can consider (\ref{eq:first-boundary-conditions}) as boundary
conditions where (\ref{eq:first-boundary-conditions-a}) is the condition
at the left endpoint and (\ref{eq:first-boundary-conditions-b}) is the
condition at the right endpoint.

The system (\ref{eq:recurrence}) with
(\ref{eq:first-boundary-conditions}), restricted to $k\in
\{1,2,\dots,N\}\setminus\{m_i\}_{i=1}^{n}$, can be solved
recursively whenever the first $n$ entries of the vector $\varphi$ are
given. Let $\varphi^{(j)}(z)$ ( $j\in\{1,\dots,n\}$) be a solution of
(\ref{eq:recurrence}) for all $k\in
\{1,2,\dots,N\}\setminus\{m_i\}_{i=1}^{n}$ such that
\begin{equation}
  \label{eq:initial-condition}
\inner{\delta_i}{\varphi^{(j)}(z)}=t_{ji}, \;\text{for}\; i=1,\dots,n\,,
\end{equation}
where $\mathscr{T}=\{t_{ji}\}_{j,i=1}^n$ is an upper triangular real
matrix and $t_{jj}\ne 0$ for all $j\in\{1,\dots,n\}$. In
(\ref{eq:initial-condition}) and in the sequel, we consider the inner
product in $\cH$ to be antilinear in its first argument.

The condition given by (\ref{eq:initial-condition}) can be seen as the
initial conditions for the system (\ref{eq:recurrence}) and
(\ref{eq:first-boundary-conditions-a}). We emphasize that given the
boundary condition at the left endpoint
(\ref{eq:first-boundary-conditions-a}) and the initial condition
(\ref{eq:initial-condition}), the system restricted to
$k\in\{1,2,\dots,N\}\setminus\{m_i\}_{i=1}^{n}$ has a unique
solution for any fixed $j\in\{1,\dots,n\}$ and $z\in\mathbb{C}$.

\begin{remark}
\label{rem:fundamental-system}
Note that the properties of the matrix $\mathscr{T}$ guarantee that
the collection of vectors $\{\varphi^{(j)}(z)\}_{j=1}^n$ is a
fundamental system of solutions of (\ref{eq:recurrence}) restricted to
$k\in \{1,2,\dots,N\}\setminus\{m_i\}_{i=1}^{n}$ with the boundary
condition (\ref{eq:first-boundary-conditions-a}).
\end{remark}
The entries of the vector $\varphi^{(j)}(z)$ are
polynomials, so we denote $P_k^{(j)}(z):=\varphi_k^{(j)}(z)$, for all
$k\in\{1,\dots,N\}$. And, define
\begin{equation*}
  Q_i^{(j)}(z):=
(z-d_{m_i}^{(0)})P_{m_i}^{(j)}(z)-\sum_{k=0}^{n-1}
d_{m_i-n+k}^{(n-k)}P_{m_i-n+k}^{(j)}(z)-\sum_{k=1}^{n-i}d_{m_i}^{(k)}P_{m_i+k}^{(j)}(z)
\end{equation*}
for $i\in\{1,\dots,n\}$ (it is assumed that the last sum is zero when
$i=n$).

It is worth remarking that the polynomials $\{P_k^{(j)}(z)\}_{k=1}^N$
and $\{Q_i^{(j)}(z)\}_{i=1}^n$ depend on the initial conditions given
by the matrix $\mathscr{T}$.

Define the matrix

\begin{equation*}
\mathscr{Q}(z):=\begin{pmatrix}
Q_{1}^{(1)}(z)&\dots&Q_{1}^{(n)}(z)\\
\vdots&\ddots&\vdots\\
Q_{n}^{(1)}(z)&\dots&Q_{n}^{(n)}(z)
\end{pmatrix}\,, \quad \text{for }z\in\complex\,.
\end{equation*}

There exists a solution of (\ref{eq:eigenvectors}) for a given
$z\in\complex$ if and only if there is a vector $(\beta_1(z),\dots,\beta_n(z))^t$ such that
\begin{equation}
  \label{eq:system-homoge}
    \mathscr{Q}(z)\begin{pmatrix}
\beta_1(z)\\
\vdots\\
\beta_n(z)
\end{pmatrix}=0\,.
\end{equation}

Indeed, since $\{\varphi^{(j)}(z)\}_{j=1}^n$ is a fundamental system
of solutions of (\ref{eq:recurrence}) restricted to
$k\in \{1,2,\dots,N\}\setminus\{m_i\}_{i=1}^{n}$ with the boundary
condition (\ref{eq:first-boundary-conditions-a}), the vector $\beta(z)$, given by
\begin{equation}
\label{eq:sol-beta}
\beta(z)=\sum_{j=1}^n\beta_j(z)\varphi^{(j)}(z)\,,
\end{equation}
is a solution of (\ref{eq:recurrence}) restricted to
$k\in \{1,2,\dots,N\}\setminus\{m_i\}_{i=1}^{n}$, satisfying
(\ref{eq:first-boundary-conditions-a}), for any collection of complex
numbers $\{\beta_j(z)\}_{j=1}^n$. Thus, it follows from
(\ref{eq:recurrence}) and (\ref{eq:first-boundary-conditions-b}) that
\begin{equation*}
  (A-zI)\beta(z)=\sum_{k=1}^{N}c_k(z)\delta_k\,,
\end{equation*}
where
\begin{equation*}
c_k(z):=
\begin{cases}
  \sum_{j=1}^{n}\beta_j(z)Q_{i}^{(j)}(z)&\text{if } k=m_i,
  \text{ for all } i=1,\dots,n\,,\\
  0&\text{otherwise}\,.
\end{cases}
\end{equation*}

Therefore, (\ref{eq:sol-beta}) is a solution of
(\ref{eq:eigenvectors}) if and only if
\begin{equation}
\label{system-beta}
    \sum_{j=1}^{n}\beta_j(z)Q_{i}^{(j)}(z)=0
\end{equation}
for all $i\in\{1,\dots,n\}$, which is equivalent to
(\ref{eq:system-homoge}).

\begin{lemma}
\label{lem:rank-equivalent-multiplicity}
Let $\widetilde{n}(z):=\dim \ker(A-zI)$. Then,
\begin{equation*}
  \rank(\mathscr{Q}(z))=n-\widetilde{n}(z)\,.
\end{equation*}
Observe that $\widetilde{n}(z)\leq n$ for all $z\in\complex$.
\end{lemma}
\begin{proof}
  The proof is straightforward. Having fixed $z\in\mathbb{C}$, one
  recurs to the Kronecker-Capelli-Rouch\'e Theorem (see
  \cite[Chap.\,3 Secs.\,1-2]{MR891403}) to obtain that the dimension of
  the space of solutions of (\ref{eq:system-homoge}) is equal to
  $n-\rank(\mathscr{Q}(z))$.
\end{proof}

Immediately from Lemma~\ref{lem:rank-equivalent-multiplicity} it
follows that
\begin{equation*}
  \spec(A)=\{z\in\mathbb{C}: \det\mathscr{Q}(z)=0\}\,.
\end{equation*}

Fix $j\in\{1,\dots,n\}$. For $\varphi^{(j)}(z_0)$ to be a solution of
(\ref{eq:eigenvectors}), the equation
\begin{equation}
  \label{eq:poly-Q-equal-zero}
  Q_{i}^{(j)}(z_0)=0
\end{equation}
should be satisfied for any $i\in\{1,\dots,n\}$. The conditions
(\ref{eq:poly-Q-equal-zero}) can be seen as \emph{inner} boundary
conditions (of the right endpoint type) for the difference equation
(\ref{eq:recurrence}). Note that the degeneration of diagonals gives
rise to inner boundary conditions.

Let $\{x_k\}_{k=1}^N$ be such that $x_k\in\spec(A)$ for $k=1,\dots,N$. Note
that the elements of this sequence have been enumerated taking
into account the multiplicity of eigenvalues. Also, let $\alpha(x_k)$ be
the corresponding eigenvectors such that
\begin{equation*}
  \inner{\alpha(x_k)}{\alpha(x_l)}=\delta_{kl}\,,\quad
  \text{ with }k,l\in\{1,\dots,N\}\,.
\end{equation*}

It follows from Remark~\ref{rem:fundamental-system} that, for any
fixed $k=1,\dots,N$, there are complex numbers $\alpha_j(x_k)$
($j=1,\dots,n$) such that
\begin{equation}
  \label{eq:linear-combination-phi-k}
  \alpha(x_k)=\sum_{j=1}^n\alpha_j(x_k)\varphi^{(j)}(x_k)
\end{equation}
for any $k\in\{1,\dots,N\}$. Clearly, by construction
\begin{equation}
  \label{eq:condition-on-alpha-misha}
\sum_{j=1}^n\abs{\alpha_j(x_k)}>0\quad\text{for all }\,
k\in\{1,\dots,N\}\,.
\end{equation}
Additionally, since $\{\alpha(x_k)\}_{k=1}^N$ is a basis of $\cH$,
it follows from (\ref{eq:initial-condition}) that
\begin{equation}
  \label{eq:condition-on-alpha}
\sum_{k=1}^N\abs{\alpha_j(x_k)}>0\quad\text{for all }\,
j\in\{1,\dots,n\}\,.
\end{equation}
By (\ref{system-beta}) and the fact that $\alpha(x_k)\in\ker(A-x_kI)$,
it follows that
\begin{equation}
  \label{eq:interpolation-alpha}
  \sum_{j=1}^n\alpha_j(x_k)Q_i^{(j)}(x_k)=
0 \quad \text{for all } i\in\{1,\dots,n\}
\end{equation}
is true.

Now, define the matrix valued function
\begin{equation}
  \label{eq:function-measure}
  \sigma(t):=\sum_{x_k<t}\sigma_k\,,
\end{equation}
where
 \begin{equation}
   \label{eq:measure-sum}
 \sigma_k=
   \left(\begin{array}{cccc}
 \abs{\alpha_1(x_k)}^2&{\alpha_1(x_k)}\overline{\alpha_2(x_k)}
&\dots&{\alpha_1(x_k)}\overline{\alpha_n(x_k)}\\
 {\alpha_2(x_k)}\overline{\alpha_1(x_k)}
&\abs{\alpha_2(x_k)}^2&\dots&{\alpha_2(x_k)}\overline{\alpha_n(x_k)}\\
 \vdots&\vdots&\ddots&\vdots\\
{\alpha_n(x_k)}\overline{\alpha_1(x_k)}
&{\alpha_n(x_k)}\overline{\alpha_2(x_k)}&\dots&\abs{\alpha_n(x_k)}^2
 \end{array}
\right)
 \end{equation}
is a rank-one, nonnegative matrix
(cf. \cite[Sec.\,1]{MR1668981}). Note that
$\sigma(t)=\sigma^{\mathscr{T}}(t)$ depends on the initial conditions
given by $\mathscr{T}$.

We have thus arrived at the following
\begin{proposition}
\label{prop:properties-spectral-function}
The matrix valued function $\sigma(t)$ has the following properties:
\begin{enumerate}[$i$)]
\item It is a nondecreasing monotone step
  function.\label{item:properties-spectral-1}
\item Each jump of the function is a matrix whose rank is not greater than
  $n$.\label{item:properties-spectral-2}
\item The sum of the ranks of all jumps is equal to $N$ (the
  dimension of the space
  $\mathcal{H}$).\label{item:properties-spectral-3}
\end{enumerate}
\end{proposition}

\begin{remark}
  \label{rem:matrix-valued-measures}
  A nondecreasing matrix valued function satisfying
  $\ref{item:properties-spectral-2}$)--$\ref{item:properties-spectral-3}$)
  uniquely determines a Borel matrix valued measure with finite
  support (see \cite[Sec.\,72]{MR1255973}).
\end{remark}

For any matrix valued function $\sigma(t)$ satisfying properties
$\ref{item:properties-spectral-1}$)-$\ref{item:properties-spectral-3}$),
there is a collection of vectors $\{\alpha(x_k)\}_{k=1}^N$ satisfying
(\ref{eq:condition-on-alpha-misha}) and
(\ref{eq:condition-on-alpha}) such that $\sigma(t)$ is given by
(\ref{eq:function-measure}) and (\ref{eq:measure-sum})
(cf. \cite[Thm.\,2.2]{MR1668981}).

If $\mathscr{T}=I$, then
$\sigma_{ij}(t)=\inner{\delta_i}{E(t)\delta_j}$
($i,j\in\{1,\dots,n\}$), where $E(t)$ is the spectral resolution of
$A$. Indeed,
\begin{align*}
  \inner{\delta_i}{E(t)\delta_j}&
  =\inner{\delta_i}{\sum_{x_l<t}\inner{\alpha(x_l)}{\delta_j}\alpha(x_l)}=
  \sum_{x_l<t}\inner{\alpha(x_l)}{\delta_j}\inner{\delta_i}{\alpha(x_l)}\\
  &=\sum_{x_l<t}\overline{\alpha_j(x_l)}\alpha_i(x_l)=\sigma_{ij}(t)\,.
\end{align*}

Therefore, in this case, the matrix valued function $\sigma(t)$ is the
spectral function of the operator $A$ with respect to
$\{\delta_k\}_{k=1}^n$.

\begin{definition}
  \label{def:sigma-mathfrak-M}
  The set of all matrix valued functions $\sigma(t)$ given by
  (\ref{eq:function-measure}) and (\ref{eq:measure-sum}), where the
  collection of vectors $\{\alpha(x_k)\}_{k=1}^N$ satisfies
  (\ref{eq:condition-on-alpha-misha}) and
  (\ref{eq:condition-on-alpha}), is denoted by $\mathfrak{M}(n,N)$.
\end{definition}

Note that any matrix valued function in $\mathfrak{M}(n,N)$ satisfies
$\ref{item:properties-spectral-1}$)--$\ref{item:properties-spectral-3}$)
of Proposition~\ref{prop:properties-spectral-function}. On the basis
of what has been discussed we refer to the matrix valued functions in
$\mathfrak{M}(n,N)$ as spectral functions. Alternatively, by
Remark~\ref{rem:matrix-valued-measures}, one can consider the
corresponding matrix valued measures (spectral measures).

Consider the Hilbert space $L_2(\mathbb{R},\sigma)$, where $\sigma$ is
the spectral function corresponding to the operator $A$ given by
(\ref{eq:function-measure}) and (\ref{eq:measure-sum}) (see
\cite[Sec.\,72]{MR1255973}). We agree that the inner product
$\inner{\cdot}{\cdot}$ is antilinear in its first argument. Clearly,
the property $\ref{item:properties-spectral-3}$) implies that
$L_2(\mathbb{R},\sigma)$ is an $N$-dimensional space and in each
equivalence class there is an $n$-dimensional vector polynomial.

Define the vector polynomials in $L_2(\mathbb{R},\sigma)$
\begin{equation}
\label{eq:def-q}
\boldsymbol{q}_i(z):=(Q_i^{(1)}(z),\dots,Q_i^{(n)}(z))^t
\end{equation}
for all $i\in\{1,\dots,n\}$, and
 \begin{equation}
 \label{eq:def-p-i}
   \boldsymbol{p}_k(z):=\left(P_k^{(1)}(z),\dots,P_k^{(n)}(z)\right)^t
 \end{equation}
 for all $k\in\{1,\dots,N\}$.
\begin{lemma}
  \label{lem:ortonormal-p-L2}
  The vector polynomials $\{\boldsymbol{p}_k(z)\}_{k=1}^N$, defined by
  (\ref{eq:def-p-i}), satisfy
\begin{equation*}
  \inner{\boldsymbol{p}_j}{\boldsymbol{p}_k}_{L_2(\mathbb{R},\sigma)}
=\delta_{jk}
\end{equation*}
for $j,k\in\{1,\dots,N\}$.
\end{lemma}
\begin{proof}
  \begin{align*}
    \inner{\boldsymbol{p}_j}{\boldsymbol{p}_k}_{L_2(\mathbb{R},\sigma)}=
&\sum_{l=1}^N\inner{\boldsymbol{p}_j(x_l)}{\sigma_l\boldsymbol{p}_k(x_l)}\\
=&\sum_{l=1}^N{\left(\sum_{s=1}^n\alpha_s(x_l)\overline{P_{j}^{(s)}(x_l)}\right)}
\sum_{s=1}^n\overline{\alpha_s(x_l)}P_{k}^{(s)}(x_l)\\
=&\sum_{l=1}^N\inner{\delta_j}{\alpha(x_l)}\inner{\alpha(x_l)}{\delta_k}
=\delta_{jk}\,,
  \end{align*}
  where it has been used that
  $\delta_k=\sum_{i=1}^{n}\inner{\alpha(x_i)}{\delta_k}\alpha(x_i)$.
\end{proof}

Let $U:\mathcal{H} \rightarrow L_2(\mathbb{R},\sigma)$ be the
isometry given by $U\delta_k\mapsto \boldsymbol{p}_k$, for all
$k\in\{1,\dots,N\}$. Under this isometry, the operator $A$ becomes the
operator of multiplication by the independent variable in
$L_2(\mathbb{R},\sigma)$. Indeed,
 \begin{align*}
   \inner{\delta_j}{A\delta_k}&=\inner{\sum_{l=1}^N\inner{\alpha(x_l)}{\delta_j}
     \alpha(x_l)}{A\sum_{s=1}^N\inner{\alpha(x_s)}{\delta_k}\alpha(x_s)}\\
   &=\sum_{l=1}^N\inner{\delta_j}{\alpha(x_l)}\inner{\alpha(x_l)}{\sum_{s=1}^Nx_s
     \inner{\alpha(x_s)}{\delta_k}\alpha(x_s)}\\
   &=\sum_{l=1}^N\inner{\delta_j}{\alpha(x_l)}x_l\inner{\alpha(x_l)}{\delta_k}\\
   &=\inner{\boldsymbol{p}_j}{t\boldsymbol{p}_k}_{L_2(\reals,\sigma)}\,.
 \end{align*}

 If the matrix $\mathscr{T}$ in (\ref{eq:initial-condition}) turns out
 to be the identity matrix, i.\,e., $\mathscr{T}=I$, then it can be
 shown that $U^{-1}$ is the isomorphism corresponding to the canonical
 representation of the operator $A$ \cite[Sec.\,75]{MR1255973}, that
 is,
 \begin{equation*}
    \delta_k= U^{-1}\boldsymbol{p}_k=\sum_{j=1}^{n}P_k^{(j)}(A)\delta_j
 \end{equation*}
for all $k\in\{1,\dots,N\}$.

\begin{remark}
    \label{rem:recurrence-equation-vector}
    The matrix representation of the multiplication operator in
    $L_2(\mathbb{R,\sigma})$ with respect to the basis
    $\{\boldsymbol{p}_1(z),\dots,\boldsymbol{p}_N(z)\}$ is again the matrix
    $\mathcal{A}$. Thus,
    \begin{equation}
      \label{eq:recurrence-equation-vector}
      \sum_{i=0}^{n-1} d_{k-n+i}^{(n-i)}
      \boldsymbol{p}_{k-n+i}(z)+d_k^{(0)}\boldsymbol{p}_k(z)+\sum_{i=1}^nd_k^{(i)}
      \boldsymbol{p}_{k+i}(z)=
      z \boldsymbol{p}_k(z)\,
    \end{equation}
for $k=1,\dots,N$, where it is assumed that $\pb{p}_l=0$ whenever $l<1$.
Also, one verifies that
\begin{equation}
  \label{eq:def-q-by-recurrence}
  \boldsymbol{q}_j(z)=(z-d_{m_j}^{(0)})\boldsymbol{p}_{m_j}(z)-
\sum_{i=0}^{n-1} d_{m_j-n+i}^{(n-i)}\boldsymbol{p}_{m_j-n+i}(z)-
\sum_{i=1}^{n-j}d_{m_j}^{(i)}\boldsymbol{p}_{m_j+i}(z)
\end{equation}
for all $j\in\{1,\dots,n\}$, where the last sum vanishes when $j=n$.
\end{remark}
The relationship between the spectral functions
$\sigma^{\mathscr{T}}=\sigma$ and $\sigma^I$ for an arbitrary
$\mathscr{T}$ is given by the following lemma.

\begin{lemma}
  \label{lem:spectral_function_relationship}
  Fix a natural number $N>n$. For any $n\times n$ upper triangular
  real matrix with no zeros in the main diagonal $\mathscr{T}$, the
  spectral function $\sigma^{\mathscr{T}}$ given in
  (\ref{eq:function-measure}) satisfies
  \begin{equation*}
    \mathscr{T}^*\sigma^{\mathscr{T}}\mathscr{T}=\sigma^I\,.
  \end{equation*}
\end{lemma}
\begin{proof}
  Let $\mathscr{T}$ be $n\times n$ upper triangular real matrix with no
  zeros in the main diagonal. Then, by
  (\ref{eq:linear-combination-phi-k}) one has

\begin{equation*}
  \inner{\delta_j}{\alpha(x_l)}=\sum_{i=1}^j\alpha_i(x_l)t_{ij}\,,
\quad \forall j\in\{1,\dots,n\}\,.
\end{equation*}

Now, for the particular case, when $\mathscr{T}=I$, one considers
\begin{equation*}
\sigma^I(t;i,j)=\sum\limits_{x_l<t}{\alpha_i'(x_l)}\overline{\alpha_j'(x_l)}\,.
\end{equation*}
Therefore,
$\inner{\delta_j}{\alpha(x_l)}_{\complex^N}=\alpha_j'(x_l)$ and
\begin{equation}
\label{eq:identity-and-general-case-t}
\sigma^I(t;i,j)=\sum_{x_l<t}{\sum_{k=1}^i\alpha_k(x_l)t_{ki}}
{\sum_{s=1}^j\overline{\alpha_s(x_l)}t_{ks}}\,.
\end{equation}
Observe that
\begin{equation*}
\mathscr{T}^*\left(
\begin{smallmatrix}
{\alpha_1(x_l)}\\ \vdots\\{\alpha_n(x_l)}
\end{smallmatrix}
\right)=\left(
\begin{smallmatrix}
  {\alpha_1(x_l)t_{11}}\\
{\sum_{k=1}^2\alpha_k(x_l)t_{k2}}\\
\vdots\\
{\sum_{k=1}^n\alpha_k(x_l)t_{kn}}
\end{smallmatrix}\right)
\,,
\end{equation*}
and by (\ref{eq:identity-and-general-case-t})
\begin{equation*}
\sigma^I(t;i,j)=\left(\mathscr{T}^*\sigma^{\mathscr{T}}\mathscr{T}\right)(t;i,j)\,.
\end{equation*}
\end{proof}

An immediate consequence of the previous lemma is the following
assertion

\begin{corollary}
  \label{cor:finite-measure-finite-first-moment}
  Fix $N>n$. For any $n\times n$ upper triangular matrix $\mathscr{T}$
  with no zeros in the main diagonal, one has
\begin{equation*}
  \mathscr{T}^*\int_{\reals}d\,\sigma^{\mathscr{T}}\mathscr{T}
=\int_{\reals}d\,\sigma^I=I\,.
\end{equation*}
\end{corollary}

\section{Connection with a linear interpolation problem}
\label{sec:direct-spectral-analysis}
Motivated by (\ref{eq:interpolation-alpha}), we consider the following
interpolation problem. Given a collection of complex numbers
$\{z_k\}_{k=1}^N$ and $\{\alpha_j(k)\}_{j=1}^n$ ($k=1,\dots,N$), find
the scalar polynomials $R_j(z)$ ($j=1,\dots,n$) which satisfy the equation
  \begin{equation}
\label{eq:introduction-interpolation-problem}
      \sum_{j=1}^n\alpha_j(k)R_j(z_k)=0\,, \quad \forall k\in\{1,\dots,N\}\,.
  \end{equation}
The polynomials satisfying
(\ref{eq:introduction-interpolation-problem}) are the solutions to the
interpolation problem and the numbers  $\{z_k\}_{k=1}^N$ are called
the interpolation nodes.

In \cite{MR3389906}, this interpolation problem is studied
in detail. Let us introduce some of the notions and results given in
\cite{MR3389906}.

\begin{definition}
  \label{def:set-S-n}
  For a collection of complex numbers $z_1,\dots,z_N$, and matrices
  $\sigma_k:=\left\{{\alpha_{i}(k)}\overline{\alpha_{j}(k)}\right\}_{i,j=1}^{n}$
  ($k\in\{1,\dots,N\}$), let us consider the equations
   \begin{equation}
     \label{eq:interpolation-problem-sigma}
     \inner{\boldsymbol{r}(z_k)}{\sigma_k\boldsymbol{r}(z_k)}_{\complex^n}=0
   \end{equation}
   for $k=1,\dots,N$, where $\boldsymbol{r}(z)$ is a nonzero
   $n$-dimensional vector polynomial. We denote by
   $\mathbb{S}=\mathbb{S}(\{\sigma_k\}_{k=1}^N,\{z_k\}_{k=1}^N)$ the
   set of all vector polynomials $\boldsymbol{r}(z)$ which satisfy
   (\ref{eq:interpolation-problem-sigma})
   (c.f. \cite[Def.\,3]{MR3389906}).
\end{definition}

It is worth remarking that solving
(\ref{eq:interpolation-problem-sigma}) is equivalent to solving the
linear interpolation problem
(\ref{eq:introduction-interpolation-problem}) whenever
$\boldsymbol{r}(z)=(R_1(z),\dots,R_n(z))^t$.
 \begin{definition}
\label{def:height}
Let $\boldsymbol{r}(z)=\left(R_1(z),R_2(z),\ldots,R_n(z)\right)^{t}$
be an $n$-dimensional vector polynomial. The height of
$\boldsymbol{r}(z)$ is the number
\begin{equation*}
h(\boldsymbol{r}):=
\max_{j\in\{1,\dots,n\}}\left\lbrace n\deg (R_j)+j-1\right\rbrace\,,
\end{equation*}
where it is assumed that $\deg 0:=-\infty$ and $h(\boldsymbol{0}):=-\infty$.
\end{definition}

In \cite[Thm.\,2.1]{MR3389906} the following proposition is
proven.
\begin{proposition}
\label{prop:basis-vector}
Let $\{\boldsymbol{g}_1(z),\dots,\boldsymbol{g}_{m+1}(z)\}$ be a
sequence of vector polynomials such that $h(\boldsymbol{g}_i)=i-1$ for
all $i\in\{1,\dots,m+1\}$. Any vector polynomial $\boldsymbol{r}(z)$
with height $m\ne-\infty$ can be written as follows
\begin{equation*}
  \boldsymbol{r}(z)=\sum_{i=1}^{m+1}c_i\boldsymbol{g}_{i}(z)\,,
\end{equation*}
where $c_i\in\mathbb{C}$  for all $i\in\{1,\dots,n\}$
and $c_{m+1}\neq0$.
\end{proposition}

\begin{definition}
  \label{def:height-set}
  Let $\mathcal{S}$ be an arbitrary subset of the set of all
  $n$-dimensional vector polynomials. We define the height of
  $\mathcal{S}$ by
\begin{equation*}
h(\mathcal{S}):=\min \left\{h(\boldsymbol{r}): \:
  \boldsymbol{r}\in\mathcal{S}, \:\boldsymbol{r}\neq 0\right\}\,.
\end{equation*}
We say that $\boldsymbol{r}(z)$ in the set $\mathbb{S}$ is a first
generator of $\mathbb{S}$ when
\begin{equation*}
h(\boldsymbol{r})=
h(\mathbb{S})\,.
\end{equation*}
\end{definition}

\begin{definition}
  \label{def:m-spaces}
  For any fixed arbitrary vector polynomial $\pb{r}$, let $\mathbb{M}(\boldsymbol{r})$ be the subset of vector
  polynomials
  given by
\begin{equation*}
  \mathbb{M}(\boldsymbol{r}):=
  \left\{\boldsymbol{s}(z): \:\boldsymbol{s}(z)=
    S(z)\boldsymbol{r}(z), \: S(z) \:
    \text{ is an arbitrary scalar polynomial}\right\}\,.
\end{equation*}
\end{definition}
Note that for all $\boldsymbol{s}(z)\in\mathbb{M}(\boldsymbol{r})$,
there is a $k\in\mathbb{N}\cup\{0\}$ such that
\begin{equation*}
  h(\boldsymbol{s})=nk+h(\boldsymbol{r})\,.
\end{equation*}
In this case $k=\deg S$, where $\boldsymbol{s}(z)=S(z)\boldsymbol{r}(z)$.

\begin{proposition}
\label{prop:form-height-S_n-minus-set-M}
(\cite[Lem.\,4.3]{MR3389906}) Fix a natural number $m$ such
that $1\le m<n$. If the vector polynomials
$\boldsymbol{r}_1(z),\dots,\boldsymbol{r}_m(z)$ are arbitrary elements
of $\mathbb{S}$, then
\begin{equation*}
  h\left(\mathbb{S}\setminus[\mathbb{M}(\boldsymbol{r}_1)+\dots+
    \mathbb{M}(\boldsymbol{r}_m)]\right)\neq
  h(\boldsymbol{r}_j)+nk
\end{equation*}
for any $j\in\{1,\dots,m\}$ and
$k\in\mathbb{N}\cup\left\{0\right\}$. In other words,
\begin{equation*}
  h\left(\mathbb{S}\setminus[\mathbb{M}(\boldsymbol{r}_1)+\dots+
    \mathbb{M}(\boldsymbol{r}_m)]\right)\quad\text{and}\quad
  h(\boldsymbol{r}_j)
\end{equation*}
are different elements of the factor space $\integers/n\integers$ for
any $j\in\{1,\dots,m\}$.
\end{proposition}

Due to Proposition \ref{prop:form-height-S_n-minus-set-M}, the
following definition makes sense.

\begin{definition}
\label{def:j-generator}
One defines recursively the $j$-th generator of $\mathbb{S}$ as the
vector polynomial $\boldsymbol{r}_j(z)$ in
$\mathbb{S}\setminus[\mathbb{M}(\pb{r}_1)
\dotplus\dots\dotplus\mathbb{M}(\pb{r}_{j-1})]$ such that
\begin{equation*}
  h(\boldsymbol{r}_j)=
h(\mathbb{S}
\setminus[\mathbb{M}(\pb{r}_1)\dotplus\dots\dotplus\mathbb{M}(\pb{r}_{j-1})])\,.
\end{equation*}
\end{definition}

In \cite[Thm.\,5.3 and Rem.\,3]{MR3389906}, the following
results were obtained.
\begin{proposition}
\label{prop:all-solution-generators}
 There are exactly $n$ generators of $\mathbb{S}$. Moreover, if the
 vector polynomials
 $\boldsymbol{r}_1(z),\dots,\boldsymbol{r}_n(z)$ are the
 generators of $\mathbb{S}$, then
 \begin{equation*}
\mathbb{S}=
\mathbb{M}(\boldsymbol{r}_1)\dotplus\dots\dotplus\mathbb{M}(\boldsymbol{r}_n)\,
 \end{equation*}
 and the heights of the generators of $\mathbb{S}$ are different
 elements of the factor space $\integers/n\integers$.
\end{proposition}

\begin{proposition}
\label{prop:sum-height-generators}
Let $\boldsymbol{r}_j(z)$ be the $j$-th generator of
$\mathbb{S}(n,N)$. It holds true that
 \begin{equation}
   \label{eq:sum-height-generators}
   \sum_{j=1}^nh(\boldsymbol{r}_j)=Nn+\frac{n(n-1)}{2}\,.
 \end{equation}
\end{proposition}
Now, let us apply these results to the spectral analysis of the
operator $A$. To this end, consider
the solution of (\ref{eq:interpolation-alpha}) as elements of
$\mathbb{S}(\{\sigma_k\}_{k=1}^N,\{x_k\}_{k=1}^N)$, where $\sigma_k$
is given by (\ref{eq:measure-sum}).

\begin{lemma}
  \label{lem:degenerate-q}
  Fix $j\in\{1,\dots,n\}$ and let $\{x_k\}_{k=1}^N=\spec(A)$, where
  $\{x_k\}_{k=1}^N$ is enumerated tacking into account the
  multiplicity of eigenvalues. If $\boldsymbol{q}_j(z)$ is the vector
  polynomial given in (\ref{eq:def-q}), then
  \begin{equation*}
    \boldsymbol{q}_j(z)\in\mathbb{S}(\{\sigma_k\}_{k=1}^N,\{x_k\}_{k=1}^N)\,.
  \end{equation*}
\end{lemma}
\begin{proof}
  The assertion follows by comparing (\ref{system-beta}) with
 (\ref{eq:interpolation-alpha}).
\end{proof}
From this lemma, taking into account the definition of the inner
product in $L_2(\reals,\sigma)$ (see the proof of
Lemma~\ref{lem:ortonormal-p-L2}) and Definition~\ref{def:set-S-n}, one
arrives at the following assertion.
\begin{corollary}
\label{coro:properties-L2-q}
For all $j\in\{1,\dots,n\}$ the vector polynomial
$\boldsymbol{q}_j(z)$ is in the equivalence class of the zero in
$L_2(\mathbb{R},\sigma)$, that is,
\begin{equation*}
  \inner{\boldsymbol{q}_j}{\boldsymbol{q}_j}_{L_2(\mathbb{R},\sigma)}=0
\end{equation*}
and, for all $\boldsymbol{r}\in\L_2(\mathbb{R},\sigma)$,
\begin{equation}
  \label{eq:class-zero-arb-p}
    \inner{\boldsymbol{r}}{\boldsymbol{q}_j}_{L_2(\mathbb{R},\sigma)}=0\,.
\end{equation}
\end{corollary}

\begin{lemma}
  \label{lem:height-recurrence}
Fix $k\in\{1,\dots,N\}$.
  \begin{enumerate}[$i$)]
  \item  If $m_{j}<k<m_{j+1}$, with
    $j=0,\dots,n-1$ and $m_0=0$, then
    \begin{equation*}
      h(\pb{p}_{k+n-j})=n+h(\pb{p}_{k})\,.
    \end{equation*}\label{item:lemma-heights-1}
  \item  If there are no degenerations of the diagonals, then
    \begin{equation*}
      h(\pb{p}_{k})=k-1\,,\quad \text{for all }k\in\{1,\dots,N\}\,.
    \end{equation*}\label{item:lemma-heights-2}
  \item For any $i\in\{1,\dots,N\}$ and $j\in\{1,\dots,n\}$,
    the following holds
    \begin{equation*}
      h(\pb{p}_i)\neq h(\pb{p}_{m_j})+n = h(\pb{q}_j)\,.
    \end{equation*}\label{item:lemma-heights-3}
  \end{enumerate}
\end{lemma}
\begin{proof}
  $\ref{item:lemma-heights-1}$) The heights of the vector polynomials
  $\{\pb{p}_k\}_{k=n+1}^N$ are determined recursively by means of the
  system (\ref{eq:recurrence-equation-vector}). For any
  $m_{j}<k<m_{j+1}$, with $j=0,\dots,n-1$, one has the equation
  \begin{equation*}
    \dots +d_k^{(0)}\pb{p}_k+d_k^{(1)}\pb{p}_{k+1}+\dots+
d_{k}^{(n-j)}\pb{p}_{k+n-j}=z\pb{p}_k\,.
  \end{equation*}
  Since $d_{k}^{(n-j)}$ never vanishes, the height of $\pb{p}_{k+n-j}$
  coincides with the one of $z\pb{p}_k$, this implies
  the assertion.

  $\ref{item:lemma-heights-2}$) If there are no degenerations of the
  diagonals, then $m_1=N-n+1$. So, the heights of the vector
  polynomials $\pb{p}_1,\dots,\pb{p}_N$ are determined by
  $\ref{item:lemma-heights-1}$) with
  $j=0$.

  $\ref{item:lemma-heights-3}$) The assertions follows from the
  recurrence equations (\ref{eq:recurrence-equation-vector}) and
  (\ref{eq:def-q-by-recurrence}).
\end{proof}

\begin{lemma}
  \label{lem:cover-all-heights}
  For any nonnegative integer $s$, there exist $k\in\{1,\dots,N\}$ or
  a pair $j\in\{1,\dots,n\}$ and $l\in\nats\cup\{0\}$ such that either
  $s=h(\pb{p}_{k})$ or $s=h(\pb{q}_{j})+nl$.
\end{lemma}
\begin{proof}
  Due to
  Lemma~\ref{lem:height-recurrence}~$\ref{item:lemma-heights-1}$), it
  follows from (\ref{eq:initial-condition}) and (\ref{eq:def-p-i})
  that
\begin{equation}
\label{eq:firs-height}
  h(\pb{p}_{k})=k-1\quad\text{for}\ k=1,\dots,h(\pb{q}_1)\,.
\end{equation}
Suppose that there is $s\in\nats$ ($s>n$) such that $s\ne h(\pb{p}_k)$
for all $k\in\{1,\dots,N\}$ and $s\ne h(\pb{q}_j)+nl$ for all
$j\in\{1,\dots,n\}$ and $l\in\nats\cup\{0\}$. Let $\hat{l}$ be an
integer such that
$s-n\hat{l}\in\{h(\pb{p}_k)\}_{k=1}^N\cup\{h(\pb{q}_j)+nl\}$
($j\in\{1,\dots,n\}$ and $l\in\nats\cup\{0\}$). There is always such
an integer due to (\ref{eq:firs-height}) and the fact that
$h(\pb{q}_1)>n$ (see
Lemma~\ref{lem:height-recurrence}~$\ref{item:lemma-heights-3}$). We
take $\hat{l}_0$ to be the minimum of all $\hat{l}$'s. Thus, there is
$k_0\in\{1,\dots,N\}$ or $j_0\in\{1,\dots,n\}$, respectively, such
that either
\begin{enumerate}[a)]
\item $s-n\hat{l}_0=h(\pb{p}_{k_0})$ or \label{item:height-p-lemma-proof}
\item $s-n\hat{l}_0=h(\pb{q}_{j_0})+nl$, with
  $l\in\nats\cup\{0\}$.  \label{item:height-q-lemma-proof}
\end{enumerate}
In the case \ref{item:height-p-lemma-proof}), we prove that
$\hat{l}_0$ is not the minimum integer, this implies the assertion of
the lemma. Indeed, if there is $j\in\{1,\dots,n\}$ such that
$k_0=m_{j}$, then
$s-n\hat{l}_0+n=h(\pb{p}_{m_{j_0}})+n=h(\pb{q}_{j_0})$ due to
  Lemma~\ref{lem:height-recurrence}~$\ref{item:lemma-heights-3}$). If
there is not such $j$, then $m_j<k_0<m_{j+1}$, and
Lemma~\ref{lem:height-recurrence}~$\ref{item:lemma-heights-1}$)
implies $s-n\hat{l}_0+n=h(\pb{p}_{k_0})+n=h(\pb{p}_{k_0+n-j})$.

For the case \ref{item:height-q-lemma-proof}), if
$s-n\hat{l}_0=h(\pb{q}_{j_0})+nl$, then
$s=h(\pb{q}_{j_0})+n(l+\hat{l}_0)$ which is a contradiction.
\end{proof}
As a consequence of Proposition~\ref{prop:basis-vector}, the above
lemma yields the following result.
\begin{corollary}
\label{cor:basis-vector}
Any vector polynomial $\boldsymbol{r}(z)$ is a finite linear
combination of
\begin{equation*}
  \{\pb{p}_k(z): k\in\{1,\dots,N\}\} \cup \{z^l\pb{q}_j(z):
l\in\nats,\, j\in\{1,\dots,n\}\}\,.
\end{equation*}
\end{corollary}

\begin{theorem}
\label{thm:q-j-generator}
For $j\in\{1,\dots,n\}$, the vector polynomial $\boldsymbol{q}_j(z)$
is a $j$-th generator of
  \begin{equation*}
    \mathbb{S}(\{\sigma_k\}_{k=1}^N,\{x_k\}_{k=1}^N)\,.
  \end{equation*}
\end{theorem}
\begin{proof}
  For any fixed $j\in\{1,\dots,n\}$, suppose that there is an element
  $\boldsymbol{r}(z)\in\mathbb{S}(\{\sigma_k\}_{k=1}^N,\{x_k\}_{k=1}^N)
\setminus\left(\mathbb{M}(\boldsymbol{q}_1)\dotplus\dots
\dotplus\mathbb{M}(\boldsymbol{q}_{j-1})\right)$,
  where $\boldsymbol{q}_0(z):=0$ such that
  $h(\boldsymbol{q}_{j-1})<h(\boldsymbol{r})<h(\boldsymbol{q}_j)$. Write
  $\boldsymbol{r}$ as Corollary~\ref{cor:basis-vector}, then by
  Corollary~\ref{coro:properties-L2-q}
\begin{align*}
  0=\inner{\boldsymbol{r}}{\boldsymbol{r}}_{L_2(\mathbb{R},\sigma)}=
\inner{\sum_{k=1}^{N}c_k\boldsymbol{p}_k}{\sum_{k=1}^{N}c_k\boldsymbol{p}_k}
_{L_2(\mathbb{R},\sigma)}=\sum_{k=1}^{N}\abs{c_k}^2\,.
\end{align*}
This implies that $c_k=0$ for all $k\in\{1,\dots,N\}$. In turn, again
by Corollary~\ref{cor:basis-vector}, one has
\begin{equation*}
  \boldsymbol{r}(z)\in\mathbb{M}(\boldsymbol{q}_1)
  \dotplus\dots\dotplus\mathbb{M}(\boldsymbol{q}_{j-1})
\end{equation*}
for $j>1$, and $\boldsymbol{r}(z)\equiv 0$ for $j=1$. This
contradiction yields that $\boldsymbol{q}_j(z)$ satisfies the
definition of $j$-generator for any $j\in\{1,\dots,n\}$.
\end{proof}
The following assertion is a direct consequence of
Theorem~\ref{thm:q-j-generator},
Proposition~\ref{prop:all-solution-generators}, and
Proposition~\ref{prop:sum-height-generators}.
\begin{corollary}
\label{coro:q-space-factor}
Let $\{\boldsymbol{q}_1(z),\dots,\boldsymbol{q}_n(z)\}$ be the
$n$-dimensional vector polynomials defined by (\ref{eq:def-q}). Then,
$h(\boldsymbol{q}_1),\dots,h(\boldsymbol{q}_n)$ are different elements
of the equivqlence class of the factor space
$\mathbb{Z}/n\mathbb{Z}$. Also,
\begin{equation*}
  \sum_{j=1}^{n}h(\boldsymbol{q}_j)=Nn+\frac{n(n-1)}{2}\,.
\end{equation*}
\end{corollary}

\section{Reconstruction}
\label{sec:reconstruction}
In this section, we take as a starting point a matrix valued function
$\widetilde{\sigma}\in\mathfrak{M}(n,N)$ and construct a matrix
$\mathcal{A}$ in $\mathcal{M}(n,N)$ from this function. Moreover, we
verify that, for some matrix $\mathscr{T}$ giving the initial
conditions, the function $\sigma$ generated by the matrix
$\mathcal{A}$ (see Section~\ref{sec:spectral-measure}) coincides with
$\widetilde{\sigma}$. Thus, the results of this section show that any
matrix in $\mathcal{M}(n,N)$ can be reconstructed from its function in
$\mathfrak{M}(n,N)$.

Let $\widetilde{\sigma}(t)$ be a matrix valued function in
$\mathfrak{M}(n,N)$. Thus, one can associate an interpolation problem
(\ref{eq:introduction-interpolation-problem}) which is equivalent to
(\ref{eq:interpolation-problem-sigma}) (with $\widetilde{\sigma}_k$
instead of $\sigma_k$).  Then, by Proposition
\ref{prop:all-solution-generators} there are $n$ generators
$\tb{q}_1(z),\dots,\tb{q}_n(z)$ of
$\mathbb{S}(\{\widetilde{\sigma}_k\}_{k=1}^{N},\{\widetilde{z}_k\}_{k=1}^N)$.


Let $\{\boldsymbol{e}_i(z)\}_{i\in\mathbb{N}}$ be a sequence of
$n$-dimensional vector polynomials defined by
\begin{equation}
  \label{eq:e-k-vectors}
  \boldsymbol{e}_{nk+1}(z):=
  \left(\begin{array}{c}z^k\\0\\0\\\vdots\\0\end{array}\right),
  \boldsymbol{e}_{nk+2}(z)
  :=\left(\begin{array}{c}0\\z^k\\0\\\vdots\\0\end{array}\right),
  \dots,\boldsymbol{e}_{n(k+1)}(z)
  :=\left(\begin{array}{c}0\\0\\\vdots\\0\\z^k\end{array}\right)\,.
\end{equation}
Clearly, $h(\boldsymbol{e}_i)=i-1$. In the Hilbert space
$L_2(\reals,\widetilde{\sigma})$, let us apply the Gram-Schmidt
procedure of orthonormalization to the sequence
$\{\boldsymbol{e}_i(z)\}_{i\in\mathbb{N}}$. Suppose that $\tb{q}_1$ is
the first generator of the corresponding interpolation problem and let
$\{\tb{p}_k\}_{k=1}^{h(\tb{q}_1)}$ be the orthonormalized vector
polynomials obtained by the first $h(\tb{q}_1)$ iterations of the
Gram-Schmidt procedure. Hence, if one defines
\begin{equation*}
  \pb{s}:=\pb{e}_{h(\tb{q}_1)+1}-\sum_{i=1}^{h(\tb{q}_1)}
  \inner{\tb{p}_i}{\pb{e}_{h(\tb{q}_1)+1}}_{L_2(\reals,\widetilde{\sigma})}\tb{p}_i\,,
\end{equation*}
then, in view of the fact that $h(\tb{p}_k)=k-1$ for $k=1,\dots,
h(\tb{q}_1)$, one has
$\pb{e}_{h(\tb{q}_1)+1}=c\tb{q}_1+\sum_{i=1}^{h(\tb{q}_1)}c_i\tb{p}_i$ which
in turn leads to
\begin{equation}
\label{eq:gram-schmidt-vector}
\pb{s}=c\tb{q}_1+\sum_{k=1}^{h(\tb{q}_1)}\widetilde{c}_k\tb{p}_k\,.
\end{equation}
This implies that
$\norm{\pb{s}}_{L_2(\reals,\widetilde{\sigma})}=0$.  One continues
with the procedure by taking the next vector of the sequence
\eqref{eq:e-k-vectors}. Note that if $\tb{p}_k$ is a normalized
element given by the Gram-Schmidt procedure, then the next
\emph{completed} iteration yields a normalized vector $\tb{p}_{k+1}$.
Observe that if the Gram-Schmidt technique has produced a vector
polynomial $\pb{q}$ of zero norm  and height $h$, then for any integer
number $l$, the vector polynomial $\pb{t}$ which is obtained at the
$h+1+nl$-th iteration of the Gram-Schmidt process, that is,
\begin{equation*}
  \pb{t}=\pb{e}_{h+1+nl}\,-\!\!\!\!\sum_{h(\tb{p}_i)<h+nl}
  \!\!\!\!\inner{\tb{p}_i}{\pb{e}_{h+1+nl}}_{L_2(\reals,\widetilde{\sigma})}\tb{p}_i\,,
\end{equation*}
satisfies that $\norm{\pb{t}}_{L_2(\reals,\widetilde{\sigma})}=0$ (for all
$l\in\nats$) due to the fact that
\begin{equation*}
  \pb{e}_{h+1+nl}=R_l\pb{q}+\!\!\!\!\sum_{h(\tb{p}_i)<h+1+nl}\!\!\!\! c_i\tb{p}_i+
  \!\!\!\!\sum_{h(\pb{r}_i)<h+1+nl} \!\!\!\!\pb{r}_i\,,
\end{equation*}
where $R_l$ is a scalar polynomial with degree $k$ and each $\pb{r}_i$
is a vector polynomial of zero norm with $h(\pb{r}_i)\ne h+nk$.

Since $L_2(\reals,\widetilde{\sigma})$ has dimension $N$, then one
obtains from the Gram-Schmidt algorithm the orthonormalized sequence
$\{\tb{p}_k\}_{k=1}^N$. Furthermore, the zeros found by the
unsuccesful iterations yields all the generators
$\{\tb{q}_i\}_{i=1}^n$ of the interpolation problem given by
$\widetilde{\sigma}$ and also polynomials in
$\mathbb{S}(\{\widetilde{\sigma}_k\}_{k=1}^{N},\{\widetilde{z}_k\}_{k=1}^N)$
such that their height are of the form $h(\tb{q}_i)+nl$ with
$i\in\{1,\dots,n\}$ and $l\in\nats$.

\begin{remark}
  \label{rem:heights-p-N-and-q-n}
  If one is interested only in the heights of the vector polynomials
  $\tb{p}_N$ and $\tb{q}_n$, then one can stop the Gram-Schmidt procedure
  when $\tb{q}_{n-1}$ appears and use
  Corollary~\ref{coro:q-space-factor}.
\end{remark}
  Due to the fact that
  \begin{equation}
    \label{eq:GramScmidt-procedure}
    h\left(\pb{e}_k-\!\!\!\!\sum_{h(\pb{p}_i)<k-1}\!\!\!\!
\inner{\tb{p}_i}{\pb{e}_k}_{L_2(\reals,\widetilde{\sigma})}\tb{p}_i\right)
    =h(\pb{e}_k)\,,
  \end{equation}
  the heights of the set
  $\{\tb{p}_k(z)\}_{k=1}^N\cup\{
  z^l\tb{q}_i(z)\}_{i=1}^n$ are in one-to-one correspondence
  with the set $\{0\}\cup\mathbb{N}$. Thus, in view of
  Proposition~\ref{prop:basis-vector}, one can write any
  $n$-dimensional vector polynomial $\boldsymbol{r}$ as
\begin{equation*}
\boldsymbol{r}(z)=\sum_{k=1}^{N}c_k\tb{p}_k(z)+
\sum_{j=1}^{n}S_j(z)\tb{q}_j(z)\,,
\end{equation*}
where $c_k\in\mathbb{C}$, $S_j(z)$ are scalar polynomials. Also,
$c_k=0$, respectively $S_j(z)=0$, if
$h(\boldsymbol{r})<h(\widetilde{\boldsymbol{p}}_k)$, respectively
$h(\boldsymbol{r})<h(\widetilde{\boldsymbol{q}}_j)$. In particular, for
$k\in\{1,\dots,N\}$,
  \begin{equation}
    \label{eq:representation-z-q}
    z\widetilde{\boldsymbol{p}}_k(z)=
\sum_{l=1}^{N}c_{lk}\widetilde{\boldsymbol{p}}_l(z)+
\sum_{j=1}^nS_{kj}(z)\tb{q}_j(z)\,,
  \end{equation}
where $c_{kl}\in\mathbb{C}$ and $S_{kj}(z)$ is scalar
polynomial.
\begin{remark}
\label{rem:basis-g}
In \eqref{eq:representation-z-q}, it holds that, for each $k$,
\begin{enumerate}[$i$)]
\item $c_{lk}=0$ if $h(z\tb{p}_k)<h(\tb{p}_l)$\label{item:heigt-equal-zp-p1},
\item $S_{kj}(z)=0$ if $h(z\tb{p}_k)<h(S_{kj}(z)\tb{q}_j)$\label{item:heigt-equal-zp-p2},
\item $c_{lk}>0$ if there is $l\in\mathbb{N}$ such that
  $h(z\tb{p}_k)=h(\tb{p}_l)$.\label{item:heigt-equal-zp-p}
\end{enumerate}
Items~$\ref{item:heigt-equal-zp-p1}$) and
$\ref{item:heigt-equal-zp-p2}$) are obtained by comparing the heights
of the left and right hand sides of (\ref{eq:representation-z-q}).  In
item~$\ref{item:heigt-equal-zp-p}$), one has to take into account that
the leading coefficient of $\pb{e}_k$ is positive for $k\in\nats$ and
therefore the Gram-Schmidt procedure yields the sequence
$\{\tb{p}_k\}_{k=1}^N$ with its elements having positive leading
coefficients.
\end{remark}
\begin{remark}
  \label{rem:hight-p_i}
  For $k=1,\dots,h_1$. We have that
  \begin{equation*}
    h(\tb{p}_k)=k-1\,,
  \end{equation*}
and, for $k>h_1$
\begin{equation*}
h(\tb{p}_k)=k-1+b_k\,,
\end{equation*}
where $b_k$ is the number of elements in the set
$\mathbb{M}(\tb{q}_1)\dotplus\dots\dotplus\mathbb{M}(\tb{q}_n)$
obtained by the Gram-Schmidt procedure and whose heights are less than
$h(\tb{p}_k)$. Observe that $b_k<b_{n+k}$, which in turn implies
\begin{equation}
\label{eq:height-p-N+i}
 h(\tb{p}_{n+k})>h(\tb{p}_{k})+n\,.
\end{equation}
\end{remark}
Therefore, if we take the inner product of
(\ref{eq:representation-z-q}) with $\widetilde{\boldsymbol{p}}_l(z)$
in $L_2(\mathbb{R},\widetilde{\sigma})$, we obtain
\begin{equation}
\label{eq:def-c-ki-inner}
c_{lk}=
\inner{\widetilde{\boldsymbol{p}}_l}
{z\widetilde{\boldsymbol{p}}_k}_{L_2(\mathbb{R},\widetilde{\sigma})}
=\inner{z\widetilde{\boldsymbol{p}_l}}
{\widetilde{\boldsymbol{p}}_k}_{L_2(\mathbb{R},\widetilde{\sigma})}
=c_{kl}\,,
\end{equation}
where (\ref{eq:class-zero-arb-p}) has been used. Hence, the matrix
$\{c_{lk}\}_{l,k=1}^N$ is symmetric and it is the matrix
representation of the operator of multiplication by the independent
variable in $L_2(\mathbb{R},\widetilde{\sigma})$ with respect to the
basis $\{\tb{p}_k(z)\}_{k=1}^N$.

The following results shed light on the structure of the matrix
$\{c_{lk}\}_{l,k=1}^N$.
\begin{lemma}
  \label{lem:band-matrix}
If $\abs{l-k}>n$. Then,
\begin{equation*}
  c_{kl}=c_{lk}=0\,.
\end{equation*}
\end{lemma}
\begin{proof}
  For $l-k>n$, we obtain from (\ref{eq:height-p-N+i}) that
  $h(\tb{p}_l)>h(\tb{p}_{k+n})\geq h(\tb{p}_k)+n=h(z\tb{p}_k)$. Therefore
  by Remark \ref{rem:basis-g},
 \begin{equation*}
 c_{lk}=\inner{\tb{p}_l}{z\tb{p}_k}_{L_2(\mathbb{R},\widetilde{\sigma})}=0\,.
 \end{equation*}
 And similarly for $k-l>n$.
\end{proof}
Lemma \ref{lem:band-matrix} shows that $\{c_{lk}\}_{l,k=1}^N$ is a
band matrix. Let us turn to the question of characterizing the
diagonals of $\{c_{lk}\}_{l,k=1}^N$. It will be shown that they
undergo the kind of degeneration given in the Introduction.

For a fixed number $i\in\{0,\dots,n\}$, we define the numbers
\begin{equation}
  \label{eq:entries-matrix-inverse-prob}
  d^{(i)}_k:=c_{k+i, k}=c_{k, k+i}
\end{equation}
for $k=1,\dots,N-i$.

\begin{lemma}
  \label{lem:reconstruction-entries}
  Fix $j\in\{0,\dots,n-1\}$.
  \begin{enumerate}[$i$)]
  \item If $k$ is such that $ h(\tb{q}_{j})<
    h(z\tb{p}_k)<h(\tb{q}_{j+1})$, then $d_k^{(n-j)}>0$. Here one
    assumes that $h(\boldsymbol{q}_0):=n-1$.
  \item If $k$ is such that $h(z\tb{p}_k)\geq h(\tb{q}_{j+1})$, one
    has that $d_k^{(n-j)}=0$.
  \end{enumerate}

\end{lemma}
\begin{proof}
  Fix a number $j\in\{0,\dots,n-1\}$, then any vector polynomial of the
  basis $\{\tb{p}_k(z)\}_{k=1}^N$ satisfies either
\begin{equation}
\label{eq:height_z_p_between_q}
h(\tb{q}_{j})<h(z\tb{p}_k)<h(\tb{q}_{j+1})
\end{equation}
or
\begin{equation}
\label{eq:height-z_p-less-q_i+1}
h(z\tb{p}_k)\geq h(\tb{q}_{j+1})\,.
\end{equation}
Suppose that
$k\in\{1,\dots,N\}$ is such that (\ref{eq:height_z_p_between_q})
holds, then there is $l\in\{1,\dots,N\}$ such that
\begin{equation*}
  h(\tb{p}_l)=h(\tb{p}_k)+n=h(z\tb{p}_{k})\,.
\end{equation*}
Indeed, if there is no vector polynomial $\tb{p}_l(z)$ such that
$h(\tb{p}_l)=h(z\tb{p}_k)$, then $h(z\tb{p}_k)=h(z^s\tb{q}_i)$ for
some $i\leq j$ and $s\geq 1$. Therefore
$h(\tb{p}_k)=h(z^{s-1}\tb{q}_i)$, which contradicts the fact that the
heights of the set $\{\tb{p}_k\}_{k=1}^N\cup\{z^l\tb{q}_j\}_{j=1}^n$
($l\in\nats\cup\{0\}$) are in one-to-one correspondence with the set
$\{0\}\cup\nats$.

Let $f_k$ be the number of elements of the sequence
$\{\tb{g}_i\}_{i=1}^\infty$ in
\begin{equation*}
\mathbb{M}(\tb{q}_1)\dotplus\dots\dotplus\mathbb{M}(\tb{q}_{i})
\end{equation*}
whose heights lies between $h(\tb{p}_k)$ and $h(\tb{p}_k)+n$. If one
assumes that (\ref{eq:height_z_p_between_q}) holds, then
\begin{equation}
\label{eq:a_k-first}
f_k=j\,.
\end{equation}
This is so because there are $n-1$ ``places'' between $h(\tb{p}_k)$
and $h(\tb{p}_k)+n$ and, for each generator $\tb{q}_j(z)$
($j\in\{1,\dots,n\}$), the heights of the elements of
$\mathbb{M}(\tb{q_j})$ fall into the same equivalence class of
$\mathbb{Z}/n\mathbb{Z}$ (see
Proposition~\ref{prop:form-height-S_n-minus-set-M}).  By
(\ref{eq:a_k-first}), one has
\begin{equation*}
h(z\tb{p}_{k})=h(\tb{p}_{k})+n=h(\tb{p}_{k+n-f_k})=h(\tb{p}_{k+n-j})\,.
\end{equation*}
Therefore, Remark~\ref{rem:basis-g}~$~\ref{item:heigt-equal-zp-p}$)
implies that $d_{k}^{(n-j)}> 0$.

Now, suppose that (\ref{eq:height-z_p-less-q_i+1}) takes place.
In this case, one verifies that
\begin{equation}
 \label{eq:a_k-second-cases}
f_k\geq j+1\,.
\end{equation}
Let $\widetilde{f}_k$ be the number of elements in
$\{\tb{p}_k(z)\}_{k=1}^N$ whose heights lies between $h(\tb{p}_k)$ and
$h(\tb{p}_k)+n$. Then
\begin{equation*}
  h(\tb{p}_{k+\widetilde{f}_k})<h(\tb{p}_k)+n\leq h(\tb{p}_{k+\widetilde{f}_k+1})\,.
\end{equation*}
Also, it follows from (\ref{eq:a_k-second-cases}) and the equality
$n-1=f_k+\widetilde{f}_k$ that
\begin{equation*}
h(\tb{p}_{k+\widetilde{f}_k+1})\leq h(\tb{p}_{k+n-j-1})<h(\tb{p}_{k+n-j})\,.
\end{equation*}
Thus $h(\tb{p}_{k})+n<h(\tb{p}_{k+n-j})$. This implies that
$\inner{\tb{p}_{k+n-1}}{z\tb{p}_k}_{L_2(\mathbb{R},\widetilde{\sigma})}=0$,
which yields that $d_{k}^{(n-j)}=0$ whenever
(\ref{eq:height-z_p-less-q_i+1}) holds.
\end{proof}

\begin{corollary}
  The matrix representation of the operator of multiplication by the
  independent variable in $L_2(\reals,\widetilde{\sigma})$ with
  respect to the basis $\{\tb{p}_k\}_{k=1}^N$ is a matrix in
  $\mathcal{M}(n,N)$.\label{cor:rep-op-multiplication}
\end{corollary}
\begin{proof}
  Taking into account (\ref{eq:entries-matrix-inverse-prob}), it
  follows from Lemma \ref{lem:band-matrix} and
  \ref{lem:reconstruction-entries} that the matrix
  $\{c_{kl}\}_{k,l=1}^N$ whose entries are given by
  (\ref{eq:def-c-ki-inner}) is in the class $\mathcal{M}(n,N)$.
\end{proof}
\begin{remark}
  \label{rem:numbers-m-j}
  Since the matrix $\{c_{kl}\}_{k,l=1}^N$ is in $\mathcal{M}(n,N)$,
  there are numbers $\{m_i\}_{i=1}^n$ associated with the matrix (see
  Introduction). This numbers can be found from
  Lemma~\ref{lem:reconstruction-entries} which tells us that a
  degeneration occurs when there exists $k\in\{1,\dots,N\}$ such that
  $h(z\tb{p}_k)=h(\tb{q}_{j+1})$ (this happens for each
  $j\in\{0,1,\dots,n-1\}$). Thus,
\begin{equation}
  \label{eq:coincide-height-zp-q}
  h(z\tb{p}_{m_i})=h(\tb{q}_i)\,,\quad \forall i\in\{1,\dots,n\}\,.
\end{equation}
\end{remark}

It is straightforward to verify that (\ref{eq:condition-on-alpha}) is
equivalent to the fact that $\boldsymbol{e}_i(z)$ is not in the
equivalence class of zero in $L(\reals,\widetilde{\sigma})$ for
$i\in\{1,\dots,n\}$. Therefore, the first $n$ elements of
$\{\tb{p}_k(z)\}_{k=1}^N$ are obtained by applying Gram-Schmidt to the
set $\{\boldsymbol{e}_i(z)\}_{i=1}^n$. Thus, if one defines
\begin{equation}
\label{eq:triangular-upper-p}
  t_{ij}:=\inner{\delta_i}{\tb{p}_{j}}_{\complex^{n}}\,,\quad\forall i,j\in\{1,\dots,n\}\,,
\end{equation}
the matrix $\mathscr{T}=\{t_{ij}\}_{i,j=1}^n$ turns out to be upper
triangular real and $t_{jj}\ne 0$ for all $j\in\{1,\dots,n\}$.  Now,
for this matrix $\mathscr{T}$ and $\mathcal{A}$ construct the
solutions $\varphi^{(j)}(z)$ satisfying
(\ref{eq:initial-condition}). Hence, the vector polynomials
$\{\boldsymbol{p}_1(z),\dots,\boldsymbol{p}_n(z)\}$ defined by
(\ref{eq:def-p-i}) satisfy (\ref{eq:triangular-upper-p}). In other
words
\begin{equation*}
 \boldsymbol{p}_j(z)=\tb{p}_j(z)\,,\quad \forall j\in\{1,\dots,n\}\,.
\end{equation*}

Consider the recurrence equation, which is obtained from
(\ref{eq:representation-z-q}), but only for the case
$\ref{item:heigt-equal-zp-p}$) of the Remark~\ref{rem:basis-g}
taking into account (\ref{eq:entries-matrix-inverse-prob}) and
Lemma~\ref{lem:reconstruction-entries}. That is,
\begin{small}
  \begin{equation}
    \label{eq:dif-eq-tilde-p}
    \begin{aligned}
      d^{(0)}_1\tb{p}_1+\dots+d^{(n)}_1\tb{p}_{n+1}&=z\tb{p}_1\\
      d^{(1)}_1\tb{p}_1+d^{(0)}_2\tb{p}_2
      +\dots+d^{(n)}_2\tb{p}_{n+2}&=z\tb{p}_2\\
      &\vdots\\
      d_{m_1-1-n}^{(n)}\tb{p}_{m_1-1-n}
      +\dots+d^{(0)}_{m_1-1}\tb{p}_{m_1-1}+d^{(1)}_{m_1-1}
      \tb{p}_{m_1}+d_{m_1-1}^{(n)}\tb{p}_{m_1-1+n}&
      =z\tb{p}_{m_1-1}\\
      \dots+d^{(0)}_{m_1+1}\tb{p}_{m_1+1}
      +d^{(1)}_{m_1+1}\tb{p}_{m_1+2}+\dots+d_{m_1+1}^{(n-1)}
      \tb{p}_{m_1+n}+S_{m_1+1,1}\tb{q}_1&=z\tb{p}_{m_1+1}\\
      &\vdots\\
      \dots+d^{(0)}_{m_2-1}\tb{p}_{m_2-1}
      +d^{(1)}_{m_2-1}\tb{p}_{m_2}+d_{m_2-1}^{(n-1)}
      \tb{p}_{m_2-2+n}+S_{m_2-1,1}\tb{q}_1&=z\tb{p}_{m_2-1}\\
      \dots+d^{(0)}_{m_2+1}\tb{p}_{m_2+1}
      +d^{(1)}_{m_2+1}\tb{p}_{m_2+2}+\dots+d_{m_2+1}^{(n-2)}
      \tb{p}_{m_2-1+n}+\sum_{i=1}^2S_{m_2+1,i}\tb{q}_i&=z\tb{p}_{m_2+1}\\
      &\vdots
    \end{aligned}
  \end{equation}
\end{small}
Since $\boldsymbol{p}_k(z)$ and $\tb{p}_k(z)$ satisfy the same
recurrence equation for any $k$ in $\{1,\dots,m_1-1+n\}$, one has
\begin{equation*}
 \boldsymbol{p}_k(z)=\tb{p}_k(z)\,,\quad \forall k\in\{1,\dots,m_1-1+n\}\,.
\end{equation*}
From the system of equations (\ref{eq:dif-eq-tilde-p}), consider the
equation containing the vector polynomial $z\tb{p}_{m_1+1}(z)$. By
comparing this equation with the corresponding one from
(\ref{eq:recurrence-equation-vector}), one concludes
\begin{equation*}
 \boldsymbol{p}_{m_1+n}(z)=\tb{p}_{m_1+n}(z)+S(z)\tb{q}_1(z)\,,
\end{equation*}
where $S(z)$ is a scalar polynomial, so $S(z)\tb{q}_1(z)$ is in the
equivalence class of zero of
$L_2(\mathbb{R},\widetilde{\sigma})$. Observe that
$h(\tb{p}_{m_1+n})>h(S\tb{q}_1)$ since
$h(\tb{p}_{m_1+n})=h(z\tb{p}_{m_1+1})$ in the equation containing
$z\tb{p}_{m_1+1}(z)$ and the height of $z\tb{p}_{m_1+1}(z)$ does not
coincide with the height of $S(z)\tb{q}_1(z)$. Recursively, for
$k>m_1+n$, one obtains the following lemma.
\begin{lemma}
  \label{lem:equivalence-p-tilde-p}
  The vector polynomials $\{\boldsymbol{p}_k(z)\}_{k=1}^N$ and
  $\{\tb{p}_k(z)\}_{k=1}^N$ defined above (see the text below
  (\ref{eq:triangular-upper-p}) and above
  Remark~\ref{rem:heights-p-N-and-q-n}, respectively) satisfy that
\begin{equation}
\label{eq:polynomial_p_inverse_problem}
 \boldsymbol{p}_k(z)=\tb{p}_k(z)+\tb{r}_k(z)
\end{equation}
for all $k\in\{1,\dots,N\}$, where $\tb{r}_k(z)$ is in the equivalence
class of zero of $L_2(\mathbb{R},\widetilde{\sigma})$ and
$h(\tb{r}_k)<h(\boldsymbol{p}_k)$.  Therefore,
 \begin{equation*}
   h(\boldsymbol{p}_k)=h(\tb{p}_k)\,,\quad \forall k\in\{1,\dots,N\}\,.
 \end{equation*}
\end{lemma}

On the other hand, for the particular case $k=m_1$,
(\ref{eq:representation-z-q}) and (\ref{eq:coincide-height-zp-q})
imply that
\begin{equation*}
  z\tb{p}_{m_1}=d_{m_1-n}^{(n)}\tb{p}_{m_1-n}+\dots+d_{m_1}^{(0)}\tb{p}_{m_1}
  +d_{m_1}^{(1)}\tb{p}_{m_1+1}+\dots
  +d_{m_1}^{(n-1)}\tb{p}_{m_1+n-1}+\gamma_1\tb{q}_1\,,
\end{equation*}
where $\gamma_1\neq0$.

In general, one verifies that for all $j\in\{1,\dots,n\}$
\begin{align*}
  z\tb{p}_{m_j}=&d_{m_j-n}^{(n)}\tb{p}_{m_j-n}+\dots
  +d_{m_j}^{(0)}\tb{p}_{m_j}+\\
  &+d_{m_j}^{(1)}\tb{p}_{m_j+1}+\dots+d_{m_j}^{(n-j)}
  \tb{p}_{m_j+n-j}+\sum_{i<j}S_{m_j,i}\tb{q}_i+\gamma_j\tb{q}_j\,,
\end{align*}
where $\gamma_j\neq0$ and $S_i(z)$ is a scalar polynomial.
Hence,
\begin{equation}
\label{eq:tilde-q}
\begin{aligned}
  \gamma_j \tb{q}_j=&\left(z-d_{m_j}^{(0)}\right)\tb{p}_{m_j}
  -\left(d_{m_j-n}^{(n)}\tb{p}_{m_j-n}
    +\dots+d_{m_j-1}^{(1)}\tb{p}_{m_j-1}+\right.\\
  & \left. +d_{m_j}^{(1)}\tb{p}_{m_j+1}+
    \dots+d_{m_1}^{(n-j)}\tb{p}_{m_j+n-j}+\sum_{i<j}S_{m_j,i}\tb{q}_i\right)
\end{aligned}
\end{equation}
for all $j\in\{1,\dots,n\}$.

Let us define the set of vector polynomials
$\{\boldsymbol{q}_1(z),\dots\boldsymbol{q}_n(z)\}$ by means of
(\ref{eq:def-q-by-recurrence}) using
$\{\boldsymbol{p}_1(z),\dots,\boldsymbol{p}_N(z)\}$, as was done in
Section~\ref{sec:spectral-measure}.
\begin{lemma}
\label{lem:equivalence-q-tilde-q}
Let $\tb{q}_j(z)$ be $j$-generator of
$\mathbb{S}(\{\widetilde{\sigma}_k\}_{k=1}^N,\{\widetilde{x}_k\}_{k=1}^N)$,
and $\boldsymbol{q}_j(z)$ be defined as above. Then
$h(\boldsymbol{q}_j)=h(\tb{q}_j)$ for all $j\in\{1,\dots,n\}$ and
  \begin{equation}
    \label{eq:dif-q-tilde-q}
    \boldsymbol{q}_j(z)=\sum_{i\leq j}S_i(z)\tb{q}_i(z)\,,\quad S_j\neq0\,,
  \end{equation}
where $S_i(z)$ are scalar polynomials.
\end{lemma}
\begin{proof}
  It follows from (\ref{eq:def-q-by-recurrence}),
  (\ref{eq:polynomial_p_inverse_problem}) and (\ref{eq:tilde-q}) that
\begin{equation}
\label{eq:tilde-q-relation-q}
  \boldsymbol{q}_j(z)=\gamma_j\tb{q}_j(z)+\tb{s}_j(z)\,,
\quad\text{ for all }j\in\{1,\dots,n\}\,,
\end{equation}
where $\tb{s}_j(z)$ is in the equivalence class of the zero of
$L_2(\mathbb{R},\widetilde{\sigma})$ and its height is strictly less
that the height of $\tb{q}_j(z)$ since, due to
(\ref{eq:coincide-height-zp-q}), the height of $\tb{q}_j(z)$ is
strictly greater than the height of any other term in the equation
with $k=m_j$ in the system (\ref{eq:representation-z-q}). Thus,
$h(\boldsymbol{q}_j)=h(\tb{q}_j)$ for all $j\in\{1,\dots,n\}$.

Equation (\ref{eq:tilde-q-relation-q}) also shows that
$\boldsymbol{q}_i(z)\in\mathbb{S}
(\{\widetilde{\sigma}_k\}_{k=1}^N,\{\widetilde{x}_k\}_{k=1}^N)$ and,
due to Proposition \ref{prop:all-solution-generators},
(\ref{eq:dif-q-tilde-q}) is satisfied.
\end{proof}


\begin{lemma}
\label{lem:inner-coincide}
Let $\boldsymbol{r}(z)$ and $\boldsymbol{s}(z)$ be any two
$n$-dimensional vector polynomials. Then,
  \begin{align*}
    \inner{\boldsymbol{r}}{\boldsymbol{s}}_{L_2(\mathbb{R},\sigma)}=
    \inner{\boldsymbol{r}}{\boldsymbol{s}}_{L_2(\mathbb{R},\widetilde{\sigma})}\,.
  \end{align*}
\end{lemma}
\begin{proof}
  Any vector polynomial $\boldsymbol{r}(z)$ can written as
  \begin{equation*}
    \boldsymbol{r}(z)=\sum_{k=1}^{N}c_k\boldsymbol{p}_k(z)
    +\sum_{j=1}^{n}S_{j}(z)\boldsymbol{q}_j(z)\,,
  \end{equation*}
where
$c_k=\inner{\boldsymbol{r}}{\boldsymbol{p}_k}_{L_2(\mathbb{R},\sigma)}$
and $S_{j}(z)$ are scalar polynomials.
Thus,
 \begin{align*}
   \inner{\boldsymbol{r}}{\tb{p}_k}_{L_2(\mathbb{R},\widetilde{\sigma})}&
   =\inner{\sum_{l=1}^{N}c_l\boldsymbol{p}_l+\sum_{j=1}^{n}S_{j}
     \boldsymbol{q}_j}{\tb{p}_k}_{L_2(\mathbb{R},\widetilde{\sigma})}\\
   &=\inner{\sum_{l=1}^{N}c_l\left(\tb{p}_l+\tb{r}_l\right)
     +\sum_{j=1}^{n}S_{j}\left(\sum_{i\leq j}S_i\tb{q}_i\right)}
   {\tb{p}_k}_{L_2(\mathbb{R},\widetilde{\sigma})}\\
   &=\inner{\sum_{l=1}^{N}c_l\tb{p}_l}{\tb{p}_k}
   _{L_2(\mathbb{R},\widetilde{\sigma})}=c_k\,.
 \end{align*}
\end{proof}

For the functions $\sigma(t)$ and $\widetilde{\sigma}(t)$ in
$\mathfrak{M}(n,N)$ consider the points $x_k$ and $\widetilde{x}_k$,
where, respectively, $\sigma(t)$ and $\widetilde{\sigma}(t)$ have
jumps $\sigma_k$ and $\widetilde{\sigma}_k$.  By definition, $k$ takes
all the values of the set $\{1,\dots,N\}$.
\begin{lemma}
\label{lem:same-jumps}
The points where the jumps of the matrix valued functions $\sigma(t)$ and
$\widetilde{\sigma}(t)$ take place coincide, i.\,e.,
  \begin{equation*}
    x_k=\widetilde{x}_k\,,\quad\text{for all } k\in\{1,\dots,N\}\,.
  \end{equation*}
\end{lemma}
\begin{proof}
  Define the $n$-dimensional vector polynomial
  \begin{equation*}
    \boldsymbol{r}(z):=\prod_{l=1}^N(z-x_l)\boldsymbol{e}_1(z)
  \end{equation*}
(see (\ref{eq:e-k-vectors})). Therefore,
\begin{equation*}
  \inner{\boldsymbol{r}}{\boldsymbol{r}}_{L_2(\mathbb{R},\sigma)}=
  \sum_{k=1}^N\inner{\boldsymbol{r}(x_k)}
  {\sigma_k\boldsymbol{r}(x_k)}_{\mathbb{C}^n}=0\,.
\end{equation*}
Now, if one assumes that
$\{\widetilde{x}_k\}_{k=1}^{N}\setminus\{x_k\}_{k=1}^{N}\ne\emptyset$,
then
\begin{equation*}
  \inner{\boldsymbol{r}}{\boldsymbol{r}}_{L_2(\mathbb{R},\widetilde{\sigma})}=
\sum_{k=1}^N\inner{\boldsymbol{r}(\widetilde{x_k})}
{\widetilde{\sigma}_k\boldsymbol{r}(\widetilde{x_k})}_{\mathbb{C}^n}>0
\end{equation*}
due to (\ref{eq:condition-on-alpha}). In view of Lemma
\ref{lem:inner-coincide} our assumption has lead to a contradiction,
so $\{\widetilde{x}_k\}_{k=1}^{N}\subset\{x_k\}_{k=1}^{N}$. Analogously,
one proves that
$\{x_k\}_{k=1}^{N}\subset\{\widetilde{x}_k\}_{k=1}^{N}$.
\end{proof}

\begin{lemma}
\label{lem:same-multiplicity}
The jumps of the matrix valued functions $\sigma(t)$ and
$\widetilde{\sigma}(t)$ coincide, namely, for all $k\in\{1,\dots,N\}$,
\begin{equation*}
  \sigma_k=\widetilde{\sigma}_k\,.
\end{equation*}
\end{lemma}
\begin{proof}
  Define, for each $i\in\{1,\dots,n\}$, the $n$-dimensional vector
  polynomial by
\begin{equation*}
  \boldsymbol{r}_{ki}(z):=
\prod_{\substack{l=1\\l\neq k}}^N(z-x_l)\boldsymbol{e}_i(z)\,.
\end{equation*}
Thus, for all $i,j\in\{1,\dots,n\}$,
\begin{align*}
  \inner{\boldsymbol{r}_{ki}}{\boldsymbol{r}_{kj}}_{L_2(\mathbb{R,\sigma})}&
  =\sum_{s=1}^N\inner{\boldsymbol{r}_{ki}(x_s)}
  {\sigma_s\boldsymbol{r}_{kj}(x_s)}_{\mathbb{C}^n}\\
  &=\sum_{s=1}^N\prod_{\substack{l=1\\l\neq k}}^N\abs{(x_s-x_l)}^2
  {\alpha_{i}(x_s)}\overline{\alpha_j(x_s)}\\
  &=\prod_{\substack{l=1\\l\neq k}}^N\abs{x_k-x_l}^2
  {\alpha_{i}(x_k)}\overline{\alpha_j(x_k)}\,.
\end{align*}
Analogously,
\begin{equation*}
  \inner{\boldsymbol{r}_{ki}}
  {\boldsymbol{r}_{kj}}_{L_2(\mathbb{R,\widetilde{\sigma}})}=
  \prod_{\substack{l=1\\l\neq k}}^N
  \abs{x_k-x_l}^2{\widetilde{\alpha}_{i}(x_k)}
  \overline{\widetilde{\alpha}_j(x_k)}\,,
\end{equation*}
where Lemma~\ref{lem:same-jumps} was used together with the fact that
the numbers $\widetilde{\alpha}_i(x_k)$ define the entries of the
matrix $\widetilde{\sigma}_k$ (see~(\ref{eq:measure-sum})). Therefore,
by Lemma (\ref{lem:inner-coincide})
\begin{equation*}
\sigma_k=\widetilde{\sigma}_k\,,\quad \text{ for all }k\in\{1,\dots,N\}\,.
\end{equation*}
\end{proof}
The above results can be summarized in the following assertion.
\begin{theorem}
  \label{thm:sigma-unique}
  Let $\widetilde{\sigma}(t)$ be an element of $\mathfrak{M}(n,N)$ and
  $\{c_{kl}\}_{k,l=1}^N\in\mathcal{M}(n,N)$ be the corresponding
  matrix that results from applying the method of reconstruction to
  the matrix valued function $\widetilde{\sigma}(t)$ (see
  (\ref{eq:representation-z-q}) and (\ref{eq:def-c-ki-inner})). If $A$
  is the operator whose matrix representation with respect to the
  basis $\{\delta_1,\dots,\delta_N\}$ in $\mathcal{H}$, is
  $\{c_{kl}\}_{k,l=1}^N$, then there is an upper triangular real
  matrix $\mathscr{T}$ with no zeros in the main diagonal such that
  the corresponding spectral function $\sigma(t)$ for the operator $A$
  coincides with $\widetilde{\sigma}(t)$.
\end{theorem}

\begin{remark}
  \label{rem:paraphrase-theorem}
  Theorem~\ref{thm:sigma-unique},
can be paraphrased as follows: A matrix-valued function
$\widetilde{\sigma}$ in $\mathfrak{M}(n,N)$ determines uniquely a
matrix in $\mathcal{M}(n,N)$ such that $\widetilde{\sigma}$ is a
spectral measure of it.
\end{remark}

Let $\mathcal{A}$ be in $\mathcal{M}(n,N)$ and $A$ be the
  corresponding operator. Denote by $V_\theta$ the unitary operator
  whose matrix representation with respect to the canonical basis is
  $\diag\{e^{i\theta_1},\dots,e^{i\theta_N}\}$ with
  $\theta_k\in[0,2\pi)$ for any $k\in\{1,\dots,N\}$.  Define
\begin{equation}
  \label{eq:unitary-family}
  B=V_\theta AV_\theta^*\,.
\end{equation}
The matrix
representation of $B$ is in $\mathcal{M}(n,N)$ if and only if
$\theta_k=\theta_{k+ln}$ for all $k\in\{1,\dots,n\}$ such that
$k+ln\leq N$ with $l$ being a positive integer. Thus, there are
various elements of the family of
unitarily equivalent matrices corresponding to
(\ref{eq:unitary-family}) in
$\mathcal{M}(n,N)$.

\section{Alternative inverse spectral methods}
\label{sec:alternative-methods}
For Jacobi matrices there are two ways of recovering the matrix from
the spectral function $\rho$. The first one is based on the fact that
the sequence of orthonormal polynomials, constructed via the
application of Gram-Schmidt procedure to the sequence of functions
$\{t^{k-1}\}_{k=1}^\infty$ in $L_2(\reals,\rho)$, determines the
entries of the matrix (see \cite[Chap.\,1, Sec.\,1 and Chap.\,4
Sec.\,2]{MR0184042} and \cite[Sec.\,1]{MR1627806}. The second method
uses the fact that the asymptotic expansion of the $m$-Weyl function
corresponding to $\rho$ yields the matrix entries
\cite[Sec.\,3]{MR1616422}. In the case of tridiagonal block matrices,
these two methods also work with some restrictions. Indeed, consider a
finite tridiagonal block matrix
\begin{equation}
  \label{eq:block-jacobi-matrix}
  \begin{pmatrix}
Q_1&B_1^*&0&\cdots&0\\
B_1&Q_2&B_2^*&\ddots&\vdots\\
0&B_2&Q_3&\ddots&0\\
\vdots&\ddots&\ddots&\ddots&B_{K-1}^*\\
0&\dots&0&B_{K-1}&Q_{K}
\end{pmatrix}\,,
\end{equation}
where $B_k$ is invertible for all $k=1,\dots,K-1$.
According to \cite[Chap.\,7 Sec.\,2.8]{MR0222718}, one recovers the
matrix entries $Q_1,\dots,Q_{K}$ and $B_1,\dots,B_{K-1}$ from a
matrix valued function obtained from the spectral function. This
corresponds to the first method outlined above. There is also an
analogue of the second method which is based on the function $M(z)$
given by \cite[Chap.\,7 Eq.\,2.63]{MR0222718} which satisfy
\begin{equation}
\label{eq:riccati}
  M(z)^{-1}=Q_1-zI-B_1\widetilde{M}(z)B_1^*\,,
\end{equation}
where $\widetilde{M}(z)$ is the function given by \cite[Chap.7
Eq.\,2.63]{MR0222718} for the tridiagonal block matrix obtained from
\eqref{eq:block-jacobi-matrix} by deleting the first block row and
block column. Equation~\eqref{eq:riccati} is the block analogue of
\cite[Eq.\,2.15]{MR1616422}. On the basis of the asymptotic behavior
of $\widetilde{M}$, one finds $Q_1$ and $B_1B_1^*$ from
(\ref{eq:riccati}). Since, in our setting, the matrix $B_1$ is upper
triangular with positive main diagonal, one can actually obtain the
entries of $B_1$ from $B_1B_1^*$. It is possible to obtain the next
matrix entries by considering (\ref{eq:riccati}) for the next
truncated matrix.

Any matrix of the class $\mathcal{M}(n,N)$ can be written as
\eqref{eq:block-jacobi-matrix} whenever $N/n=K$. Note that if a matrix
in $\mathcal{M}(n,N)$ undergoes degeneration, then there is $k_0$ such
that $B_k$ is not invertible for all $k=k_0,\dots,K-1$.  Thus, the
methods cited above can be used for the inverse spectral analysis of
the elements of $\mathcal{M}(n,N)$ which, do not undergo degenerations
and for which $N/n\in\nats$.

The procedure developed in Section~\ref{sec:reconstruction} is
applicable to the whole class $\mathcal{M}(n,N)$, which shows that it
is more general than the methods described above. In the
reconstruction technique of Section~\ref{sec:reconstruction},
degenerations can be treated on the basis of the solution of the
linear interpolation problem for $n$-dimensional vector polynomials.

\appendix \section{Mass-spring systems}
\label{sec:Mass-spring}

This appendix briefly describes how Newton's laws of motion and the
Hooke law yield a finite difference equation which can be written by a
finite band symmetric matrix.

Consider the finite mass-spring system given by Figure~\ref{fig:3},
where we have assumed that $N$ is even.
\begin{figure}[h]
\begin{center}
\resizebox{1\textwidth}{!}{%
\begin{tikzpicture}
  [mass1/.style={circle,draw=black!80,fill=black!13,thick,inner sep=0pt,
   minimum size=5mm},
   mass2/.style={circle,draw=black!80,fill=black!13,thick,inner sep=0pt,
   minimum size=3.7mm},
   mass3/.style={circle,draw=black!80,fill=black!13,thick,inner sep=0pt,
   minimum size=5.7mm},
   mass4/.style={circle,draw=black!80,fill=black!13,thick,inner sep=0pt,
   minimum size=5mm},
   mass5/.style={circle,draw=black!80,fill=black!13,thick,inner sep=0pt,
   minimum size=4mm},
   mass6/.style={circle,draw=black!80,fill=black!13,thick,inner sep=0pt,
   minimum size=5.2mm},
   mass9/.style={circle,draw=black!80,fill=black!13,thick,inner sep=0pt,
   minimum size=5.4mm},
   massn/.style={circle,draw=black!80,fill=black!13,thick,inner sep=0pt,
   minimum size=5.2mm},
   wall/.style={postaction={draw,decorate,decoration={border,angle=-45,
   amplitude=0.3cm,segment length=1.5mm}}},
   wall1/.style={postaction={draw,decorate,decoration={border,angle=45,
   amplitude=0.3cm,segment length=1.5mm}}}]
  \node (massn) at (12.75,1) [massn] {\scriptsize$m_{N}$};
  \node (mass9) at (11.5,1) [mass9] {\scriptsize$m_{N-1}$};
  \node (mass6) at (7.75,1) [mass6] {\scriptsize$m_6$};
  \node (mass5) at (6.5,1) [mass5] {\scriptsize$m_5$};
  \node (mass4) at (5.25,1) [mass4] {\scriptsize$m_4$};
  \node (mass3) at (4.0,1) [mass3] {\scriptsize$m_3$};
  \node (mass2) at (2.75,1) [mass2] {\scriptsize$m_2$};
  \node (mass1) at (1.5,1) [mass1] {\scriptsize$m_1$};
\draw[decorate,decoration={coil,aspect=0.4,segment
  length=1.1mm,amplitude=0.7mm}] (0.5,1)  --  node[above=2pt]
  {\scriptsize$k_1$} (mass1);
\draw[decorate,decoration={coil,aspect=0.4,segment
  length=1.4mm,amplitude=0.7mm}] (mass1) -- node[above=2pt]
{\scriptsize$k_2$} (mass2);
\draw[decorate,decoration={coil,aspect=0.4,segment
  length=1.5mm,amplitude=0.7mm}] (mass2) -- node[above=2pt]
{\scriptsize$k_3$} (mass3);
\draw[decorate,decoration={coil,aspect=0.4,segment
  length=1.1mm,amplitude=0.7mm}] (mass3) -- node[above=2pt]
{\scriptsize$k_4$} (mass4);
\draw[decorate,decoration={coil,aspect=0.4,segment
  length=0.9mm,amplitude=0.7mm}] (mass4) -- node[above=2pt]
{\scriptsize$k_5$} (mass5);
\draw[decorate,decoration={coil,aspect=0.4,segment
  length=1.4mm,amplitude=0.7mm}] (mass5) -- node[above=2pt]
{\scriptsize$k_6$} (mass6);
\draw[decorate,decoration={coil,aspect=0.4,segment
  length=1.7mm,amplitude=0.7mm}] (mass6) -- node[above=2pt]
{\scriptsize$k_7$} (9,1);
\draw[decorate,decoration={coil,aspect=0.4,segment
  length=1.3mm,amplitude=0.7mm}] (10.25,1) -- node[above=2pt]
{\scriptsize$k_{N-1}$} (mass9);
\draw[decorate,decoration={coil,aspect=0.4,segment
  length=1.7mm,amplitude=0.7mm}] (mass9) -- node[above=2pt]
{\scriptsize$k_{N}$} (massn);
\draw[decorate,decoration={coil,aspect=0.4,segment
  length=0.8mm,amplitude=0.7mm}] (massn) -- node[above=2pt]
{\scriptsize$k_{N+1}$} (13.95,1);
\draw[decorate,decoration={coil,aspect=0.4,segment
  length=0.8mm,amplitude=0.7mm}] (mass1) to [bend left=60] node[above=2pt]
{\scriptsize$k'_2$} (mass3);
\draw[decorate,decoration={coil,aspect=0.4,segment
  length=1.5mm,amplitude=0.7mm}] (mass3) to [bend left=60] node[above=2pt]
{\scriptsize$k'_4$} (mass5);
\draw[decorate,decoration={coil,aspect=0.4,segment
  length=1.3mm,amplitude=0.7mm}] (mass5) to [bend left=60] node[above=2pt]
{\scriptsize$k'_6$} (9,1.3);
\draw[decorate,decoration={coil,aspect=0.4,segment
  length=0.7mm,amplitude=0.7mm}] (10.25,1.7) to [bend left=60] node[above=2pt]
{\scriptsize$k'_{N-2}$} (mass9);
\draw[decorate,decoration={coil,aspect=0.4,segment
  length=1.5mm,amplitude=0.7mm}] (mass9) to [bend left=60] node[above=2pt]
{\scriptsize$k'_{N}$} (13.95,1.25);
\draw[decorate,decoration={coil,aspect=0.4,segment
  length=1.5mm,amplitude=0.7mm}] (0.5,0.8) to [bend right=60] node[below=2pt]
{\scriptsize$k'_1$} (mass2);
\draw[decorate,decoration={coil,aspect=0.4,segment
  length=1.8mm,amplitude=0.7mm}] (mass2) to [bend right=60] node[below=2pt]
{\scriptsize$k'_3$} (mass4);
\draw[decorate,decoration={coil,aspect=0.4,segment
  length=1.2mm,amplitude=0.7mm}] (mass4) to [bend right=60] node[below=2pt]
{\scriptsize$k'_5$} (mass6);
\draw[decorate,decoration={coil,aspect=0.4,segment
  length=1.8mm,amplitude=0.7mm}] (mass6) to [bend right=60] node[below=2pt]
{\scriptsize$k'_7$} (9,.3);
\draw[decorate,decoration={coil,aspect=0.4,segment
  length=1.1mm,amplitude=0.7mm}] (10.25,.7) to [bend right=60] node[below=2pt]
{\scriptsize$k'_{N-1}$} (massn);
\draw[line width=.8pt,loosely dotted] (9.25,1) -- (10,1);
\draw[line width=.8pt,loosely dotted] (9.6,1.3) -- (10,1.7);
\draw[line width=.8pt,loosely dotted] (9.25,.3) -- (9.6,.7);
\draw[line width=.5pt,wall](0.5,1.7)--(0.5,0.3);
\draw[line width=.5pt,wall1](13.95,1.7)--(13.95,0.3);
\end{tikzpicture}
}
\end{center}
\caption{Mass-spring system of a matrix in
  $\mathcal{M}(2,N)$: nondegenerated case}\label{fig:3}
\end{figure}

In Figure~\ref{fig:3}, $m_j$ and $k_j$, $k_j'$ stand, respectively, for
the $j$-th mass, the $j$-th spring constant connecting immediate
neighbors, and the $j$-th spring constant connecting mediated
neighbors.

Due to the Hooke law, the forces $F_i$ acting on the masses $m_i$ are
given by
\begin{equation*}
  F_i=k'_{i+1}x_{i+2}+k_{i+1}x_{i+1}-(k_{i+1}+k'_{i+1}+k_i
  +k'_{i-1})x_{i}+k_{i}x_{i-1}+k'_{i-1}x_{i-2}\,.
\end{equation*}
This system of equations, due to Newton's second law, can be written
as
\begin{equation}
\label{eq:newton-law}
  M\ddot{x}=Kx\,,
\end{equation}
where
\begin{equation*}
 M= \left(\begin{smallmatrix}
  m_1&&&\\
  &m_2&&\\
&&\ddots&\\
&&&m_N
  \end{smallmatrix}\right)
\ K= \left(\begin{smallmatrix}
  \alpha_1&k_2&k'_2&0&\dots&0\\
  k_2&\alpha_2&k_3&k'_3&\dots&0\\
k'_2&k_3&\alpha_3&k_4&\ddots&0\\
0&k'_3&k_4&\alpha_4&&k'_{N-1}\\
\vdots&&\ddots&&\ddots&k_N\\
0&\dots&0 &k'_{N-1}&k_N&\alpha_N
  \end{smallmatrix}\right)
\end{equation*}
with $\alpha_i=-(k_{i+1}+k'_{i+1}+k_i+k'_{i-1})$. The system
\eqref{eq:newton-law} is equivalent to $\ddot{U}=LU$, where
$U=M^{1/2}X$ and
\begin{align*}
  L&=M^{-1/2}KM^{-1/2}\\ &=\left(\begin{smallmatrix}
      \frac{\alpha_1}{m_1}&\frac{k_2}{\sqrt{m_1m_2}}&
      \frac{k'_2}{\sqrt{m_1m_3}}&0&\dots&0\\
      \frac{k_2}{\sqrt{m_1m_2}}&\frac{\alpha_2}{m_2}&
      \frac{k_3}{\sqrt{m_2m_3}}&\frac{k'_3}{\sqrt{m_2m_4}}&\dots&0\\
      \frac{k'_2}{\sqrt{m_1m_3}}&\frac{k_3}{\sqrt{m_2m_3}}&
      \frac{\alpha_3}{m_3}&\frac{k_4}{\sqrt{m_3m_4}}&\ddots&0\\
      0&\frac{k'_3}{\sqrt{m_2m_4}}&\frac{k_4}{\sqrt{m_3m_4}}&
      \frac{\alpha_4}{m_4}&&\frac{k'_{N-1}}{\sqrt{m_{N-2}m_N}}\\
      \vdots&&\ddots&&\ddots&\frac{k_N}{\sqrt{m_{N-1}m_N}}\\
      0&\dots&0 &\frac{k'_{N-1}}{\sqrt{m_{N-2}m_N}}&
      \frac{k_N}{\sqrt{m_{ N-1}m_N}}&\frac{\alpha_N}{m_N}
  \end{smallmatrix}\right)\,.
\end{align*}
Thus, according to our notation, the diagonals are given by
\begin{align}
  d^{(0)}_j&=-\frac{k_{j+1}+k'_{j+1}+k_j+k'_{j-1}}{m_j}\label{eq:diag-1}\\
  d^{(1)}_j&=\frac{k_{j+1}}{\sqrt{m_{j+1}m_j}}\label{eq:diag-2}\\
  d^{(2)}_j&=\frac{k'_{j+1}}{\sqrt{m_{j+2}m_j}}\label{eq:diag-3}\,.
\end{align}
The eigenvalues of this matrix determine the frequencies of the
harmonic oscillations whose superposition yields the movement of the
mechanical system.

For Jacobi matrices, viz. when the masses are connected only with
their immediate neighbor, it is possible to give a finite continued
fraction which yields the quotients $k_j/m_j$ for any $j$ from the
quotient $k_1/m_1$ \cite[Rem.\,11]{MR3113459} (see also
\cite[pag.\,76]{mono-marchenko}). This reconstruction is physically
meaningful. In the general case, one can construct the following
continued fractions from \eqref{eq:diag-1}, \eqref{eq:diag-2}, and
\eqref{eq:diag-3}. Note that the first equation reduces to the
continued fraction of \cite[Rem.\,11]{MR3113459} when $k'_j=0$.

\begin{align}
  \label{eq:continued-fraction}
  \frac{k_{j+1}+k'_j}{m_{j+1}}&=\frac{\left(d^{(1)}_j\right)^2
    +\sqrt{\frac{m_{j+2}}{m_j+1}}d^{(1)}_jd^{(2)}_j
    +\sqrt{\frac{m_j-1}{m_j}}d^{(2)}_{j-1}d^{(1)}_j
    +\sqrt{\frac{m_{j-1}m_{j+2}}{m_{j+1}m_j}}
    d^{(2)}_jd^{(2)}_{j-1}}{d^{(0)}_j+\frac{k_j+k'_{j-1}}{m_j}}\\
  &=\frac{\left(d^{(1)}_j\right)^2+\frac{k'_{j+1}}{k_{j+1}}
    \left(d^{(1)}_j\right)^2 +\frac{k_{j-1}}{k'_{j-1}}
    \frac{d^{(1)}_jd^{(2)}_{j-1}d^{(2)}_{j-2}}{d^{(1)}_{j-2}}
    +\frac{k'_{j+1}k_{j-1}}{k_{j+1}k'_{j-1}}
    \frac{d^{(1)}_jd^{(2)}_{j-1}d^{(2)}_{j-2}}
    {d^{(1)}_{j-2}}}{d_j^{(0)}+\frac{k_j+k'_{j-1}}{m_j}}\,.
\end{align}

{\bf Acknowledgments.} We thank the
referee for drawing our attention to
\cite{MR3087910,marchenko-slavin} and for comments leading to an
improved presentation of this work.

\def\cprime{$'$} \def\lfhook#1{\setbox0=\hbox{#1}{\ooalign{\hidewidth
  \lower1.5ex\hbox{'}\hidewidth\crcr\unhbox0}}} \def\cprime{$'$}


\begin{thebibliography}{10}

\bibitem{MR0184042}
N.~I. Akhiezer.
\newblock {\em The classical moment problem and some related questions in
  analysis}.
\newblock Translated by N. Kemmer. Hafner Publishing Co., New York, 1965.

\bibitem{MR1255973}
N.~I. Akhiezer and I.~M. Glazman.
\newblock {\em Theory of linear operators in {H}ilbert space}.
\newblock Dover Publications Inc., New York, 1993.
\newblock Translated from the Russian and with a preface by Merlynd Nestell,
  Reprint of the 1961 and 1963 translations, Two volumes bound as one.

\bibitem{MR2043894}
B.~Beckermann and A.~Osipov.
\newblock Some spectral properties of infinite band matrices.
\newblock {\em Numer. Algorithms}, 34(2-4):173--185, 2003.
\newblock International Conference on Numerical Algorithms, Vol. II (Marrakesh,
  2001).

\bibitem{MR0222718}
J.~M. Berezans{\cprime}ki{\u\i}.
\newblock {\em Expansions in eigenfunctions of selfadjoint operators}.
\newblock Translated from the Russian by R. Bolstein, J. M. Danskin, J. Rovnyak
  and L. Shulman. Translations of Mathematical Monographs, Vol. 17. American
  Mathematical Society, Providence, R.I., 1968.

\bibitem{MR629608}
F.~W. Biegler-K{\"o}nig.
\newblock Construction of band matrices from spectral data.
\newblock {\em Linear Algebra Appl.}, 40:79--87, 1981.

\bibitem{MR2263317}
M.~T. Chu and G.~H. Golub.
\newblock {\em Inverse eigenvalue problems: theory, algorithms, and
  applications}.
\newblock Numerical Mathematics and Scientific Computation. Oxford University
  Press, New York, 2005.

\bibitem{MR504044}
C.~de~Boor and G.~H. Golub.
\newblock The numerically stable reconstruction of a {J}acobi matrix from
  spectral data.
\newblock {\em Linear Algebra Appl.}, 21(3):245--260, 1978.

\bibitem{MR2915295}
R.~del Rio and M.~Kudryavtsev.
\newblock Inverse problems for {J}acobi operators: {I}. {I}nterior mass-spring
  perturbations in finite systems.
\newblock {\em Inverse Problems}, 28(5):055007, 18, 2012.

\bibitem{MR2998707}
R.~del Rio, M.~Kudryavtsev, and L.~O. Silva.
\newblock Inverse problems for {J}acobi operators {III}: {M}ass-spring
  perturbations of semi-infinite systems.
\newblock {\em Inverse Probl. Imaging}, 6(4):599--621, 2012.

\bibitem{MR3113459}
R.~del Rio, M.~Kudryavtsev, and L.~O. Silva.
\newblock Inverse problems for {J}acobi operators {II}: {M}ass perturbations of
  semi-infinite mass-spring systems.
\newblock {\em Zh. Mat. Fiz. Anal. Geom.}, 9(2):165--190, 277, 281, 2013.

\bibitem{MR1045318}
M.~G. Gasymov and G.~S. Guse{\u\i}nov.
\newblock On inverse problems of spectral analysis for infinite {J}acobi
  matrices in the limit-circle case.
\newblock {\em Dokl. Akad. Nauk SSSR}, 309(6):1293--1296, 1989.

\bibitem{MR1616422}
F.~Gesztesy and B.~Simon.
\newblock {$m$}-functions and inverse spectral analysis for finite and
  semi-infinite {J}acobi matrices.
\newblock {\em J. Anal. Math.}, 73:267--297, 1997.

\bibitem{MR2102477}
G.~M.~L. Gladwell.
\newblock {\em Inverse problems in vibration}, volume 119 of {\em Solid
  Mechanics and its Applications}.
\newblock Kluwer Academic Publishers, Dordrecht, second edition, 2004.

\bibitem{MR2494240}
L.~Golinskii and M.~Kudryavtsev.
\newblock Inverse spectral problems for a class of five-diagonal unitary
  matrices.
\newblock {\em Dokl. Akad. Nauk}, 423(1):11--13, 2008.

\bibitem{MR2533388}
L.~Golinskii and M.~Kudryavtsev.
\newblock Rational interpolation and mixed inverse spectral problem for finite
  {CMV} matrices.
\newblock {\em J. Approx. Theory}, 159(1):61--84, 2009.

\bibitem{MR2592784}
L.~Golinskii and M.~Kudryavtsev.
\newblock An inverse spectral theory for finite {CMV} matrices.
\newblock {\em Inverse Probl. Imaging}, 4(1):93--110, 2010.

\bibitem{MR0447294}
L.~J. Gray and D.~G. Wilson.
\newblock Construction of a {J}acobi matrix from spectral data.
\newblock {\em Linear Algebra and Appl.}, 14(2):131--134, 1976.

\bibitem{MR499269}
G.~{\v{S}}. Guse{\u\i}nov.
\newblock The determination of the infinite {J}acobi matrix from two spectra.
\newblock {\em Mat. Zametki}, 23(5):709--720, 1978.

\bibitem{MR0221315}
R.~Z. Halilova.
\newblock An inverse problem.
\newblock {\em Izv. Akad. Nauk Azerba\u\i d\v zan. SSR Ser. Fiz.-Tehn. Mat.
  Nauk}, 1967(3-4):169--175, 1967.

\bibitem{MR0213379}
H.~Hochstadt.
\newblock On some inverse problems in matrix theory.
\newblock {\em Arch. Math. (Basel)}, 18:201--207, 1967.

\bibitem{MR0382314}
H.~Hochstadt.
\newblock On the construction of a {J}acobi matrix from spectral data.
\newblock {\em Linear Algebra and Appl.}, 8:435--446, 1974.

\bibitem{MR549425}
H.~Hochstadt.
\newblock On the construction of a {J}acobi matrix from mixed given data.
\newblock {\em Linear Algebra Appl.}, 28:113--115, 1979.

\bibitem{MR891403}
V.~A. Ilyin and {\`E}.~G. Poznyak.
\newblock {\em Linear algebra}.
\newblock ``Mir'', Moscow, 1986.
\newblock Translated from the Russian by Irene Aleksanova.

\bibitem{MR1668981}
M.~Kudryavtsev.
\newblock The direct and the inverse problem of spectral analysis for
  five-diagonal symmetric matrices. {I}.
\newblock {\em Mat. Fiz. Anal. Geom.}, 5(3-4):182--202, 1998.

\bibitem{MR1699440}
M.~Kudryavtsev.
\newblock The direct and the inverse problem of spectral analysis for
  five-diagonal symmetric matrices. {II}.
\newblock {\em Mat. Fiz. Anal. Geom.}, 6(1-2):55--80, 1999.

\bibitem{MR3389906}
M.~Kudryavtsev, S.~Palafox, and L.~O. Silva.
\newblock On a linear interpolation problem for {$n$}-dimensional vector
  polynomials.
\newblock {\em J. Approx. Theory}, 199:45--62, 2015.

\bibitem{MR3543793}
M.~Kudryavtsev, S.~Palafox, and L.~O. Silva.
\newblock Inverse spectral analysis for a class of infinite band symmetric
  matrices.
\newblock {\em J. Math. Anal. Appl.}, 445(1):762--783, 2017.

\bibitem{MR3087910}
Y.~I. Lyubarskii and V.~A. Marchenko.
\newblock Inverse problem for small oscillations.
\newblock In {\em Spectral analysis, differential equations and mathematical
  physics: a festschrift in honor of {F}ritz {G}esztesy's 60th birthday},
  volume~87 of {\em Proc. Sympos. Pure Math.}, pages 263--290. Amer. Math.
  Soc., Providence, RI, 2013.

\bibitem{mono-marchenko}
V.~A. Marchenko.
\newblock {\em Introduction to the theory of inverse problems of spectral
  analysis}.
\newblock Universitetski Lekcii. Akta, Kharkov, 2005.
\newblock In Russian.

\bibitem{marchenko-slavin}
V.~A. Marchenko and V.~Slavin.
\newblock {\em Inverse problems in the theory of small oscilation}.
\newblock Naukova Dumka, Kiev, 2015.

\bibitem{MR636029}
M.~P. Mattis and H.~Hochstadt.
\newblock On the construction of band matrices from spectral data.
\newblock {\em Linear Algebra Appl.}, 38:109--119, 1981.

\bibitem{MR1463594}
P.~Nylen and F.~Uhlig.
\newblock Inverse eigenvalue problem: existence of special spring-mass systems.
\newblock {\em Inverse Problems}, 13(4):1071--1081, 1997.

\bibitem{MR1436689}
P.~Nylen and F.~Uhlig.
\newblock Inverse eigenvalue problems associated with spring-mass systems.
\newblock In {\em Proceedings of the {F}ifth {C}onference of the
  {I}nternational {L}inear {A}lgebra {S}ociety ({A}tlanta, {GA}, 1995)}, volume
  254, pages 409--425, 1997.

\bibitem{MR1247178}
Y.~M. Ram.
\newblock Inverse eigenvalue problem for a modified vibrating system.
\newblock {\em SIAM J. Appl. Math.}, 53(6):1762--1775, 1993.

\bibitem{MR2305710}
L.~O. Silva and R.~Weder.
\newblock On the two spectra inverse problem for semi-infinite {J}acobi
  matrices.
\newblock {\em Math. Phys. Anal. Geom.}, 9(3):263--290 (2007), 2006.

\bibitem{MR2438732}
L.~O. Silva and R.~Weder.
\newblock The two-spectra inverse problem for semi-infinite {J}acobi matrices
  in the limit-circle case.
\newblock {\em Math. Phys. Anal. Geom.}, 11(2):131--154, 2008.

\bibitem{MR1627806}
B.~Simon.
\newblock The classical moment problem as a self-adjoint finite difference
  operator.
\newblock {\em Adv. Math.}, 137(1):82--203, 1998.

\bibitem{MR2110489}
S.~M. Zagorodnyuk.
\newblock Direct and inverse spectral problems for {$(2N+1)$}-diagonal,
  complex, symmetric, non-{H}ermitian matrices.
\newblock {\em Serdica Math. J.}, 30(4):471--482, 2004.

\bibitem{MR2432761}
S.~M. Zagorodnyuk.
\newblock The direct and inverse spectral problems for {$(2N+1)$}-diagonal
  complex transposition-antisymmetric matrices.
\newblock {\em Methods Funct. Anal. Topology}, 14(2):124--131, 2008.

\end{thebibliography}
\end{document}